\documentclass[11pt]{article}

\usepackage{sn-preamble}
\usepackage{microtype}
\usepackage{mathtools} 
\usepackage{multirow} 

\newcommand{\cost}{\mathrm{cost}}
\newcommand{\costavg}{\mathrm{avg}\text{-}\mathrm{cost}}
\newcommand{\costworst}{\mathrm{worst}\text{-}\mathrm{cost}}
\newcommand{\ind}{\mathrm{ind}} 
\newcommand{\istar}{i^\ast}
\newcommand{\opt}{\mathrm{opt}}

\newcommand{\rmd}{\mathrm{d}}

\newcommand{\optavg}{\opt^{\mathrm{avg}}}
\newcommand{\optworst}{\opt^{\mathrm{w}}}
\newcommand{\error}{\mathrm{error}}
\newcommand{\bias}{\mathrm{bias}}
\newcommand{\1}[1]{\mathbf{1}\left\{#1\right\}}

\newcommand{\agrees}{\triangleleft}
\newcommand{\ex}[2]{\Ex_{#1}\left[#2\right]}
\newcommand{\pr}[2]{\Prx_{#1}\left[#2\right]}

\newcommand{\IPRR}{\textsc{IPRR}}
\newcommand{\WIPRR}{\textsc{Warmup}\text{-}\textsc{IPRR}}

\newcommand{\bpi}{\boldsymbol{\pi}}
\newcommand{\fullAlg}{\textsc{Online\text{-}Query}}

\title{
No Price Tags?~No Problem:~Query Strategies\\ for Unpriced Information 
}

\author{
Shivam Nadimpalli
\thanks{MIT. Email: \url{shivamn@mit.edu}.}
\and 
Mingda Qiao
\thanks{University of Massachusetts Amherst. Email: \url{mingda.qiao.cs@gmail.com}.}
\and 
Ronitt Rubinfeld
\thanks{MIT. Email: \url{ronitt@csail.mit.edu}.} 
}

\date{\today}

\hypersetup{
    colorlinks = true,
    linkcolor={purple},
    urlcolor={purple},
    citecolor={blue!80!black}
}

\begin{document}

\pagenumbering{gobble}

\maketitle

\hypersetup{linkcolor={purple}}

\begin{abstract}
    The classic \emph{priced query model} introduced by Charikar et al.~(STOC 2000) captures the task of computing a known function on an unknown input, where each input variable can only be revealed by paying an associated cost. 
    The goal is to design a query strategy that determines the function's value while minimizing the total cost incurred. 
    All prior work in this model, however, assumes complete advance knowledge of these query costs---an assumption that breaks down dramatically in many realistic settings. 

    We introduce a variant of the priced query model designed explicitly to handle \emph{unknown} variable costs. 
    We prove a separation from the traditional priced query model, showing that uncertainty in variable costs imposes an unavoidable overhead for every query strategy. 
    Despite this, we are able to give query strategies that essentially match our lower bound and are competitive with the best-known cost-aware strategy for an arbitrary Boolean function. 
    Our results build on a recent connection between priced query strategies and the analysis of Boolean functions, and also draw on techniques from online algorithms. 
\end{abstract}

\newpage

\pagenumbering{arabic}

\section{Introduction}
\label{sec:intro}

Information is rarely free, and its value is often revealed only after the cost to acquire it has been paid. 
A scientist may devote weeks to running an experiment, only to find the results inconclusive; a data analyst may purchase costly annotations that add little insight; a policymaker may commission new surveys, learning only afterward which questions truly mattered.
Even in familiar settings---medical diagnostics or crowd-sourced data collection---each ``query'' carries a cost and a payoff that remain unknown until a response has been obtained.   
The same tension arises in collective decision-making: for instance, the U.S.~Electoral College can be viewed as a ``majority of majorities'' function, where confident prediction demands focusing effort on the few states whose outcomes remain uncertain. 
Additionally, costly observations may sometimes yield no new information, whereas a single inexpensive one can sometimes decisively influence the outcome. 

The same fundamental difficulty runs through all these scenarios: an overall outcome depends on many unknown variables---not all of which may need to be learned---yet revealing each variable carries a cost that is itself unknown in advance.
The challenge is to allocate effort strategically: which variables to reveal, and in what order, when each becomes known only once its cost threshold has been met. 
To capture this problem, we introduce a theoretical model of information acquisition under unknown query costs.
Our framework extends the classical \emph{priced query} model of Charikar, Fagin, Guruswami, Kleinberg and Raghavan~\cite{charikar2000query} to an ``online'' setting, where the cost required to reveal a variable is hidden until the cumulative investment surpasses its threshold. 
We begin by briefly recalling the classical model.

\paragraph{Background: Query Strategies for Priced Information.} 

The \emph{priced query} model introduced in~\cite{charikar2000query} provides a clean abstraction for decision-making when query costs are known in advance.
In this model, the goal is to evaluate a Boolean function $f\isazofunc$ on an unknown input $x\in\zo^n$, where each input bit can only be revealed by paying a specified cost. 
The challenge is to design query strategies that minimize the total cost incurred by the algorithm. 

Over the past two decades, there has been substantial work designing priced query strategies for various function classes; an incomplete list includes~\cite{cicalese2005new,cicalese2005optimal,cicalese2008function,cicalese2011beyond,cicalese2011competitive,cicalese2011threshold,deshpande2014approximation,allen2017evaluation,gkenosis2018stochastic,hellerstein2018stochastic,blanc2021query,hellerstein2022adaptivity,hellerstein2024quickly,ghuge2024nonadaptive} among many others. 
The work of Charikar et al.~\cite{charikar2000query} also inspired priced variants of other basic algorithmic tasks such as learning~\cite{kaplan2005learning}, optimization~\cite{singla2018price}, and search or sorting~\cite{gupta2001sorting,kannan2003selection}. 
The problem has also been extensively studied within the operations research community under the label of \emph{sequential testing}~\cite{unluyurt2004sequential}.

\paragraph{This Work: Query Strategies for \emph{Unpriced} Information.} 

The classical model, however, assumes that the algorithm knows all query costs in advance—an assumption that fails dramatically in many natural settings. 
In practice, the cost of uncovering information is often uncertain: we discover how expensive a variable is only through the process of revealing it.
This brings us to the focus of the present work:  
\begin{center}
    Can we design effective query strategies for \emph{unpriced} information?
\end{center} 
In other words, we consider the task of designing query strategies in the setting when the variable costs are \emph{entirely unknown} to the algorithm. 
We study this problem in the framework of competitive analysis, where our goal is to design query strategies whose total cost is comparable to that of the optimal algorithm that knows the variable costs in advance. 
Throughout, we refer to this setting as the ``online'' model, to distinguish it from the ``offline'' (known-cost) framework of~\cite{charikar2000query}.

We note that our problem bears resemblance to several models in the online algorithms literature, such as threshold queries and Pandora's box. 
These online frameworks incorporate uncertainty and sequential revelation, yet they typically do not aim to evaluate a fixed Boolean function; the emphasis there is often on regret or reward objectives rather than computing a prescribed function of the inputs. 
Conversely, the classical priced-query model assumes complete knowledge of all costs in advance.
(We refer the reader to \Cref{subsec:related-work} for a detailed discussion of related work.) 
Our work bridges these two areas, both conceptually and technically.

Algorithmically, we draw on ideas from the query-strategies literature---particularly influence-based techniques from the analysis of Boolean functions, building on a recent connection highlighted by Blanc, Lange, and Tan~\cite{blanc2021query}---as well as tools from online algorithms, including martingale arguments and testing/consistency checks.
In addition to our algorithmic contributions, we establish lower bounds for both specific strategies and general online query algorithms. 
We believe that the ideas and techniques developed here will find broader applications across these domains.

\subsection{Our Results}
\label{subsec:our-results}

For the sake of readability, we restrict ourselves to our main algorithmic results and lower bounds in the discussion below. 
We summarize all our results in~\Cref{table:all-results}, and give a detailed technical overview in \Cref{sec:proof-overview}.

\subsubsection{Problem Setup}
\label{subsubsec:setup-intro}

We start with an informal overview of our setup and refer the reader to~\Cref{subsec:problem-setup} for a precise formulation.
Throughout, we assume that the goal is to compute a Boolean function $f\isazofunc$ on an unknown input $x\in\zo^n$. 
We also denote the fixed (but unknown) vector of costs as $c \in \R_{\geq 0}^n$. 
At each time step, the algorithm may invest in a single variable; its cumulative investment is maintained as a vector $\theta \in \R^n$. 
The variable $x_i$ is revealed to the algorithm when $\theta_i \geq c_i$. 

We consider both \emph{zero-error} algorithms (i.e.~algorithms that correctly compute $f(x)$ on \emph{every} input $x\in\zo^n$) as well as \emph{$\eps$-error} algorithms (i.e.~algorithms which are allowed to err in computing $f(\bx)$ for $\eps$-fraction of $\bx\sim\zo^n$). 
We benchmark our algorithms in terms of both the optimal \emph{average cost} and the \emph{worst-case cost} among all offline algorithms. 
For example, $\optavg_0 := \optavg_0(f,c)$ will denote the optimal expected cost among all zero-error offline algorithms computing $f(\bx)$ on a uniformly random $\bx\sim\zo^n$; we correspondingly define $\optavg_\eps$, $\optworst_0$, and $\optworst_\eps$ with the latter two being \textbf{\underline{w}}orst-case costs. 
See~\Cref{def:benchmark} for formal definitions of these performance measures, and note for now that
\begin{equation} \label{eq:lego}
    \optavg_\eps \leq \optavg_0 \leq \optworst_0  
    \qquad\text{and}\qquad 
    \optavg_\eps \leq \optworst_\eps \leq \optworst_0\,. 
\end{equation}
Thus, benchmarking against $\optavg_\eps$ is the strongest guarantee we can hope for. 

Much of our focus will be on the average-case setting; we will thus refer to the \emph{competitive ratio} against the benchmark \smash{$\opt^{\star}_{\diamond}$} as the ratio between the expected cost of algorithm on a uniformly random $\bx\sim\zo^n$ and the benchmark \smash{$\opt^{\star}_{\diamond}$} where $\star,\diamond$ will always be clear from context. 

\begin{table}[]
    \renewcommand{\arraystretch}{1.6}
    \centering
    \begin{tabular}{@{}lllll@{}}
    \toprule
    Function $f\isazofunc$   & Error              & Benchmark      & Competitive Ratio  \\ \midrule
    $\AND$  (\Cref{prop:and-lower-bound})                             & $0$                & $\optavg_0$    &  $\Omega(n)$                            \\
    $\Tribes$ (\Cref{thm:tribes-lower-bound-informal})                          & $\eps \in [0,1/4)$ & $\optavg_0$    &   $\Omega(\log n) $                          \\
    \midrule
    General (\Cref{thm:warmup-iprr-informal})                             & $\eps$             & $\optworst_0$  &              $O\left(\frac{\TInf[f] \log n}{\eps}\right)$             \\
    Symmetric (\Cref{thm:symmetric-informal})                           & $\eps$             & $\optavg_0$    &     $O\left(\log\frac{1}{\eps}\right)$             \\
    General (\Cref{thm:iprr-informal})                          & $O(\eps)$             & $\optavg_\eps$ & $O\left(\frac{\TInf[f] \log n}{\eps^3}\cdot\log\pbra{\frac{\log \optavg_\eps}{\eps}} \right)$ &            \\
    D.T.~with avg.~depth $d$ (\Cref{thm:follow-the-tree-general-informal}) & $\eps$             & $\optavg_0$    &                    $O\left(\frac{d}{\eps}\right)$          \\
    \bottomrule
    \end{tabular}
    \caption{An overview of our bounds on the competitive ratio in the online setting.
    The upper bounds hold for specific online algorithms that satisfy the corresponding error bounds, and D.T.~is shorthand for decision tree. 
    See also~\Cref{table:influence-bounds} for bounds on the total influence $\TInf[f]$ (cf.~\Cref{def:influence}) for some function classes.  
    The lower bounds hold against \emph{all} online algorithms with the corresponding error guarantees.} 
    \label{table:all-results}
\end{table}

\subsubsection{Separating the Offline and Online Settings} 

We first show that the online setting is strictly harder than the offline setting of~\cite{charikar2000query}. 
In particular, the following result gives a lower bound on the competitive ratio of average-case online algorithms against the benchmark $\optavg_0$, even when the online algorithm is allowed to err on a constant fraction of inputs.

\begin{theorem}[Informal version of \Cref{thm:tribes-lower-bound}] \label{thm:tribes-lower-bound-informal}
    Let $f\isazofunc$ be the $\Tribes$ function on $n$ variables (cf.~\Cref{def:tribes-instance}). 
    Then every $0.24999$-error online algorithm computing $f$ has competitive ratio $\Omega(\log n)$ against $\optavg_0$. 
\end{theorem}

The $\Tribes$ function is a well-known extremal function in Boolean function analysis and social choice theory~\cite{BenOrLinial:85,KKL:88,odonnell-book}. 
Note that that a lower bound on the competitive ratio against $\optavg_0$ is a stronger guarantee than comparing against $\optavg_\eps$ (cf.~\Cref{eq:lego}).
En route to proving~\Cref{thm:tribes-lower-bound-informal}, we give an $\Omega(n)$ lower bound on the competitive ratio against $\optavg_0$ for zero-error online algorithms computing the $\AND$ function (\Cref{prop:and-lower-bound}).

\subsubsection{Online Query Algorithms} 

Despite \Cref{thm:tribes-lower-bound-informal}, we give online algorithms that match (up to a logarithmic factor) the performance of the best known offline strategy for a \emph{generic} Boolean function $f\isazofunc$ obtained by Blanc, Lange, and Tan~\cite{blanc2021query}. 
Prior to~\cite{blanc2021query}, much of the work on priced query strategies focused on designing algorithms tailored to structured classes of Boolean functions such as disjunctions, DNF formulas, game trees, halfspaces, and so on. 
In contrast, Blanc, Lange, and Tan~\cite{blanc2021query} give an offline algorithm whose performance can be expressed in terms of a basic complexity measure of the function $f$, namely its \emph{total influence}~$\TInf[f]$:

\begin{restatable}
{definition}
{infDef}
\label{def:influence}
	Given a Boolean function $f:\zo^n\to\zo$, we define the \emph{influence of the $i^\text{th}$ variable on $f$} as 
	\[
		\Inf_i[f] := \Prx_{\bx\sim\zo^n}\sbra{f(\bx)\neq f(\bx^{\oplus i})}
	\]
	where $x^{\oplus i} = (x_1, \dots, x_{i-1}, 1-x_i, x_{i+1}, \dots, x_n)$. 
	Furthermore, we define the \emph{total influence of $f$} as 
	\[
		\TInf[f] := \sumi \Inf_i[f]. 
	\]
\end{restatable}

We note that both the coordinate influences and total influences are central notions across a number of fields including combinatorics, complexity theory, social choice, and statistical physics; see~\cite{Kalai:04,Juknacomplexitybook,odonnell-book,garban2014noise} for more information. 
See~\Cref{table:influence-bounds} for bounds on the total influence of structured function classes including DNF formulas, circuits, intersections of halfspaces, and monotone functions.

\begin{table}[]
    \centering
    \renewcommand{\arraystretch}{1.5}
    \begin{tabular}{@{}lll@{}}
    \toprule
    Family of functions $f\isazofunc$~~          & $\TInf[f]$ & Reference \\ \midrule
    Monotone functions    &  $O(\sqrt{n})$                  & \cite{BshoutyTamon:96}                                     \\
    Intersections of $s$ halfspaces           &               $O(\sqrt{n\log s})$            &   \cite{Kane:sens14}        \\
    Size-$s$ DNFs                             &                    $O(\log s)$       &    \cite{Boppana1997}       \\
    Depth-$d$ size-$s$ $\mathrm{AC}^0$ circuits         &       $O(\log s)^{d-1}$                    &    \cite{LMN:93,Boppana1997}       \\
    \bottomrule
    \end{tabular}
    \caption{Upper bounds on the total influence of various function classes.}
    \label{table:influence-bounds}
\end{table}

\paragraph{From~\cite{blanc2021query} to Unknown Costs.} 

We start by recalling \cite{blanc2021query}'s guarantee in the offline setting: 

\begin{theorem}[Theorem~1 of~\cite{blanc2021query}] \label{thm:BLT}
    Let $f\isazofunc$, $c \in \N^n$ be a \emph{known} cost vector, and $\eps \in (0, 0.5)$. 
    There is an efficient $O(\eps)$-error query strategy for $f$ with expected cost at most 
    \[
        \optavg_\eps\cdot \frac{\TInf[f]}{\eps^2}\,.
    \]
\end{theorem}

The algorithm of Blanc, Lange, and Tan~\cite{blanc2021query} establishing~\Cref{thm:BLT} follows a simple but powerful strategy: repeatedly query the variable $x_i$ that maximizes the ratio of its influence $\Inf_i[f]$ to its (known) cost $c_i$, updating the function with its restrictions as coordinates are revealed.
A natural first attempt in the \emph{unknown-cost} setting is to adapt this rule by maintaining a cumulative investment vector $\theta$ proportional to the vector of coordinate influences. 

This serves as a warm-up to our main algorithm: we show that this natural online adaptation is competitive with respect to the weakest benchmark, $\optworst_0$, up to a multiplicative factor of the total influence; see~\Cref{thm:warmup-iprr} for a formal statement. 
Even in this preliminary setting, our analysis already departs from~\cite{blanc2021query}: we require a martingale argument to control the evolution of partial influences as costs are gradually revealed.

However, recall from~\Cref{eq:lego} that $\optworst_0$ is the weakest benchmark we can hope to be competitive against. 
If we wish to compete against the stronger quantities $\optavg_\eps$ or even $\optavg_0$, this specific ``influence-proportional'' query strategy \emph{provably fails}. 
In particular, we construct instances where its competitive ratio is $\TInf[f]\cdot\wt{\Omega}(\sqrt{n})$, polynomially worse than the offline guarantee achieved by~\cite{blanc2021query} (cf.~\Cref{thm:BLT}). 
We refer the reader to~\Cref{prop:warmup-iprr-hard-instance} for the formal statement and proof, and to~\Cref{fig:address-func} for an illustration of the hard instance, which is based on a modification of the classical ``address function'' from Boolean function analysis~\cite{odonnell-book}. 

\paragraph{Our Main Algorithmic Result.} 

Our main result builds on the online adaptation of~\cite{blanc2021query}’s algorithm, while taking into account the specific structure of the hard instance discussed above.
In essence, the $\poly(n)$ degradation relative to \Cref{thm:BLT} arises from the algorithm's \emph{local} stopping condition: 
the algorithm described above terminates once the function becomes $\eps$-close to constant. 
While this stopping condition suffices for benchmarking against $\optworst_0$, we show that it fails to be competitive against the benchmark $\optavg_\eps$, as local termination can lead to significant over-investment when the input comes from a small but ``difficult'' subset of the hypercube, whereas the optimal $\eps$-error strategy could safely ignore or ``give up'' on this region. 

This motivates our main algorithmic contribution: we continue to use influence as the guiding statistic but replace the local stopping rule with a more \emph{global} termination condition that aggregates progress across coordinates.
The modification may appear minor, but its analysis is substantially more intricate: in addition to the martingale argument from the warm-up algorithm, it requires additional ``testing'' and ``consistency'' steps. 
We give a detailed technical overview of these ideas in~\Cref{sec:proof-overview} and state our main result below.

\begin{theorem}[Informal version of \Cref{thm:iprr-main} and \Cref{prop:thresholded-IPRR-on-correct-cost}]\label{thm:iprr-informal}
    Let $f\isazofunc$, $\eps \in (0, 0.5)$, and $c \in \R^n_{\geq 0}$ be an \emph{unknown} cost vector.
    There is an efficient algorithm \fullAlg~(cf.~\Cref{alg:full-alg}) which computes $f$ to error $O(\eps)$ and an expected cost of
    \[
        \optavg_\eps \cdot O\left(\frac{\TInf[f]\log n}{\eps^3} \cdot \log\pbra{\frac{\log \optavg_\eps}{\eps}}\right).
    \]
\end{theorem}

\Cref{thm:iprr-informal} thus extends the guarantee of~\cite{blanc2021query} to the online setting, matching its performance up to logarithmic factors. 
We note that in both \Cref{thm:BLT} and \Cref{thm:iprr-informal}, efficiency is measured with respect to the given representation of the function $f: \zo^n \to \zo$. 

\paragraph{Beyond the General Case.}

Beyond~\Cref{thm:iprr-informal}, which holds for arbitrary Boolean functions, we also design online query strategies for two basic and well-studied families of functions in the query-strategies literature: \emph{symmetric functions}, and functions computed by \emph{low-depth decision trees}.
We start with our guarantee for the former: 

\begin{theorem}[Informal version of~\Cref{thm:symmetric}]
\label{thm:symmetric-informal}
    For any symmetric function $f\isazofunc$, $\WIPRR(f, \eps)$ (\Cref{alg:warmup-iprr}) is an efficient $\eps$-error algorithm with an expected cost of
    \[
        \optavg_0 \cdot O\pbra{\log{\frac{1}{\eps}}}.
    \]
\end{theorem}

We note that many natural Boolean functions including the $\AND$, $\OR$, and $\MAJ$ functions are symmetric. 
Indeed, previous works~\cite{gkenosis2018stochastic,hellerstein2018stochastic,hellerstein2022adaptivity,hellerstein2024quickly} have specifically designed query strategies for symmetric functions. 
We also note that the $O(\log(1/\eps))$ competitive ratio cannot be improved in general: when $\eps = 2^{-(n+1)} < 2^{-n}$, the algorithm $\WIPRR(f, \eps)$ is a zero-error algorithm and computes any symmetric function with an expected cost of $\optavg_0 \cdot O(n)$. 
(The algorithm \WIPRR{} is deterministic, so its error probability on $\bx \sim \zo^n$ must be a multiple of $2^{-n}$.) 
This matches our lower bound for the AND function (\Cref{prop:and-lower-bound}).

We also design query strategies for functions which can be represented as a decision tree with small average depth. 
We consider the setting where the decision tree is given to the query algorithm, as well as the setting where the small average depth decision tree is unknown. 

\begin{theorem}[Informal version of \Cref{thm:follow-the-tree-general} and \Cref{remark:unknown-tree}]\label{thm:follow-the-tree-general-informal}
    There is an $\eps$-error algorithm that, given a decision tree representation of $f$ with average depth $d$, has expected cost 
    \[
        \optavg_0 \cdot \frac{d}{\eps}.
    \]
    Moreover, only assuming that $f$ has a decision tree representation of average depth $d$, there is an $\eps$-error algorithm that, without knowing $d$ or the decision tree in advance, runs in time $\poly(n)\cdot (d/\eps)^{O(d/\eps)}$ and gives an expected cost of
    \[
        \optavg_0 \cdot O\left(\frac{d}{\eps}\right).
    \]
\end{theorem}

In general, the bound given by~\Cref{thm:follow-the-tree-general-informal} is incomparable to the $\optworst_0 \cdot O((\TInf[f] \log n)/ \eps)$ bound guaranteed by~\Cref{thm:warmup-iprr-informal}. 
Note that~\Cref{thm:follow-the-tree-general-informal} is stronger in that it competes with the average-case benchmark $\optavg_0$, yet the $d/\eps$ competitive ratio can be either higher or lower than $(\TInf[f]\log n)/\eps$. 
Indeed, note that while the bound $d \ge \TInf[f]$ holds in general, both $d = \TInf[f]$ (e.g., parity functions) and $d \gg \TInf[f]$ (e.g., $\Tribes$) can be true.

\subsection{Related Work}
\label{subsec:related-work}

We briefly survey related work in both the query-strategies and online-algorithms literature.
In addition, our notion that a variable is revealed once the cumulative investment in it exceeds an unknown cost threshold bears a superficial resemblance to ``threshold queries’’ studied in differential privacy and related areas~\cite{bun2017make,cohen2024lower}, though the connection appears to be largely informal. 

\paragraph{Query Strategies for Priced Information.}
The work of Charikar et al.~\cite{charikar2000query} measured the performance of a query strategy on an instance-by-instance basis. 
Specifically, for each unknown input $x\in\zo^n$, they compare the cost incurred by an algorithm $\calA$ to the minimal cost of a certificate for $f(x)$.\footnote{A \emph{certificate} for $f(x)$ corresponds to a subset of coordinates $S\sse[n]$ whose values, when fixed according to $x$, completely determine the value of the function as $f(x)$.} 
Their goal, then, is to design algorithms that minimize this ratio for all $x\in\zo^n$. 
Their approach thus gives a worst-case instance-by-instance guarantee; in contrast, we compare the \emph{expected cost} of $\calA$ on a uniformly random input $\bx\sim\zo^n$ to the optimal expected cost. 
Charikar et al.~\cite{charikar2000query} give query strategies for functions computable by $\AND$-$\OR$ trees, with subsequent works considering other function classes (e.g.~game trees, monotone functions, and threshold functions) as well as related algorithmic tasks such as search and sorting with prices. A partial and incomplete list includes~\cite{gupta2001sorting,kannan2003selection,cicalese2005new,cicalese2005optimal,cicalese2008function,cicalese2011competitive,cicalese2011threshold,cicalese2011beyond}, among many others. 

The work of Kaplan, Kushilevitz, and Mansour~\cite{kaplan2005learning} considers query strategies within learning models with attribute costs. 
Their problem formulation can be viewed as an average-case version of that of~\cite{charikar2000query}, where an algorithm's performance is measured by comparing the expected cost of a strategy to the expected cost of the cheapest certificate for $f(\bx)$. 
Notably,~\cite{kaplan2005learning} consider arbitrary distributions on $\zo^n$ but restrict their attention to studying functions represented by disjunctions, CDNF formulas, and read-once DNF formulas. 

Our problem formulation is inspired by that of the \emph{Stochastic Boolean Function Evaluation} problem considered by Deshpande, Hellerstein, and Keletenik~\cite{deshpande2014approximation}, with a long line of subsequent work~\cite{bach2018submodular, allen2017evaluation, hellerstein2018stochastic,blanc2021query,grammel2022algorithms,hellerstein2022adaptivity,hellerstein2024quickly,ghuge2024nonadaptive}. 
As in our setup, the work of~\cite{deshpande2014approximation} measures the performance of an algorithm according to the expected cost on a random input $\bx$ drawn from a (known) product distribution on $\zo^n$. 
However, much of the focus has been on designing algorithms tailored to structured function classes such as CDNF formulas, decision trees, and linear threshold functions.

The recent work of Blanc, Lange, and Tan~\cite{blanc2021query} considers the setup of~\cite{deshpande2014approximation}, but differs from previous results in two notable ways: first, they allow the algorithm to err on a small fraction of inputs; and second, they design algorithms that apply to \emph{all} Boolean functions, with performance guarantees in terms of the total influence of the function. 
Our work aligns with theirs in both of these aspects. 

A key distinction of our work is that all previous algorithms for querying priced information---including those listed above---assume that the costs of revealing input variables are known to the algorithm beforehand.

\paragraph{Online Algorithms} 

Our work also connects to the broader literature on online algorithms. 
In particular, the problem we consider shares conceptual similarities with both classical and contemporary models such as Pandora's Box~\cite{weitzman1978optimal} and the Markovian price-of-information~\cite{gupta2019markovian}. 
A growing body of recent work~\cite{singla2018price,chawla2020pandora,qiao2023online,drygala2023online} further develops related ideas in online decision-making under cost constraints, and the literature in this area is by now too extensive to survey comprehensively. 
However, as mentioned earlier, the objectives in these models are typically \emph{reward-based}---aimed at maximizing expected value---rather than computing a specific, known function of an unknown input. 

\subsection{Discussion}
\label{subsec:discussion}


Much remains to be understood even within the online priced query model introduced in this paper. 
A substantial body of work has designed query strategies tailored to specific classes of Boolean functions in the offline setting. 
This includes monotone functions~\cite{cicalese2005optimal}, $\AND$-$\OR$ and game trees~\cite{charikar2000query,cicalese2005optimal}, subclasses of DNFs~\cite{kaplan2005learning,deshpande2014approximation,allen2017evaluation}, threshold functions and their generalizations~\cite{deshpande2014approximation,gkenosis2018stochastic,ghuge2024nonadaptive}, and symmetric functions~\cite{gkenosis2018stochastic,hellerstein2018stochastic,hellerstein2024quickly}. 
In particular, additional structural information about the target function can often be effectively leveraged to design better query strategies in the offline setting. 
It is natural to ask whether similar improvements can be obtained in the online setting. 

We further highlight two promising directions for future investigation:
\begin{itemize}

    \item \textbf{Partial Information Settings}: The model considered in this paper assumes that the algorithm begins with no knowledge of the variable costs. 
    However, a natural extension is to consider \emph{partial information} settings, where the algorithm has access to some side information about the costs---perhaps through a distributional prior, historical data, or external advice. 
    Exploring how such partial knowledge affects the design and performance of online query strategies could likely yield connections to Bayesian optimization, bandit algorithms, and learning with advice. 

    \item \textbf{Dynamic Costs:} In many real-world scenarios the variable costs may evolve dynamically, driven by e.g.~external market forces or time. 
    In such settings, the algorithm must reason not only about \emph{which} queries to make, but also about \emph{when} to make them. 
    We believe that modeling such settings and designing effective query strategies with shifting prices pose interesting challenges, and are likely to reveal new connections to online learning and decision-making under uncertainty.  

\end{itemize}

\subsection{Organization}
\label{subsec:organization}

We give a technical overview of our results in~\Cref{sec:proof-overview}. 
After recalling preliminaries in~\Cref{sec:prelims}, we prove~\Cref{thm:iprr-informal,thm:symmetric-informal} in~\Cref{sec:iprr}. 
We then give our decision tree-based query strategies establishing \Cref{thm:follow-the-tree-general-informal} in~\Cref{sec:follow-the-tree}. 
Finally, we prove separations between the offline and online setting (\Cref{thm:tribes-lower-bound-informal}) in~\Cref{sec:tribes-lb}.

\section{Technical Overview}
\label{sec:proof-overview}

Our starting point is the greedy algorithm of Blanc, Lange, and Tan~\cite{blanc2021query} for the offline setting, where the variable costs are known to the algorithm. Their key idea is to use the \emph{influence} of a variable (\Cref{def:influence}) as a proxy for its ``importance'' in determining the value of function: the algorithm keeps querying the variable with the highest influence-to-cost ratio, until a certain stopping condition is met.

\paragraph{An Online Algorithm: Influence-Proportional Round Robin.} When the costs are unknown, we use a natural strategy: \emph{invest in the variables in proportion to their influences}. 
At each step, we examine the restricted function $f_{\pi}$ induced by the input bits that have been revealed so far. 
We identify the variable $x_i$ that maximizes the ratio
\[
    \frac{\Inf_i[f_\pi]}{\theta_i},
\]
where $\Inf_i[f_\pi]$ is the influence of $x_i$ with respect to $f_\pi$, and $\theta_i$ denotes the current investment in $x_i$. Then, we increment $\theta_i$ and, if $x_i$ gets revealed as a result, we update $f_\pi$ with the knowledge of $x_i$. In other words, we keep investing in the most under-invested variable relative to its influence, maintaining the invariant that the investment in each variable is roughly proportional to its influence. 
We term this query strategy ``Influence-Proportional Round Robin'' (IPRR).

The only remaining design choice is a stopping criterion. We start with a simple one: Let $\eps$ be the target error probability. The algorithm terminates whenever $f_{\pi}$ is $\eps$-close to a constant function, at which point it outputs the rounding of $\ex{\bx\sim\zo^n}{f_{\pi}(\bx)}$ to $\{0, 1\}$. We call the resulting algorithm \WIPRR{}.

\begin{algorithm}
    \caption{The $\WIPRR$ Algorithm (Informal version of \Cref{alg:warmup-iprr})}
    \label{alg:warmup-iprr-informal}

    \vspace{1em}
    \textbf{Input:} Succinct representation of $f:\zo^n\to\zo$, error parameter $\eps \in (0, 0.5]$ \\[0.25em]
    \textbf{Output:} A bit $b \in \zo$

    \ 

    $\WIPRR(f, \eps)$:
    \begin{enumerate}
        \item Initialize $\theta \leftarrow 0^n$ and $\pi \leftarrow \emptyset$. 
        \item While $\bias(f_\pi) > \eps$:
        \begin{enumerate}
            \item Let $i^\ast \in [n]$ be the index such that $i^\ast = \arg\max_i \frac{\Inf_i[f_{\pi}]}{\theta_i}\,$.
            \item Increment $\theta_{i^\ast}$. If $x_{i^\ast}$ is revealed to be $b_{i^\ast} \in \zo$, update $\pi \leftarrow \pi \cup \{x_{i^\ast} \mapsto b_{i^\ast}\}$.
        \end{enumerate}
        \item Output $\displaystyle\mathbf{1}\bigg\{\ex{\bx\sim\zo^n}{f_\pi(\bx)} \ge 1/2\bigg\}$.
    \end{enumerate}   
\end{algorithm} 

We establish the following guarantee on the performance of \WIPRR{}:

\begin{theorem}[Informal version of \Cref{thm:warmup-iprr}]\label{thm:warmup-iprr-informal}
    The algorithm $\WIPRR(f, \eps)$ (cf.~\Cref{alg:warmup-iprr}) is an $\eps$-error algorithm for computing $f\isazofunc$ with an expected cost of
    \[
        \optworst_0 \cdot O\left(\frac{\TInf[f] \cdot \log n}{\eps}\right).
    \]
\end{theorem}

The analysis of \WIPRR{} in \Cref{thm:warmup-iprr-informal} is quite different from that of the greedy algorithm of~\cite{blanc2021query}. Roughly speaking, the proof of~\cite{blanc2021query} uses the total influence of the restricted function as a measure of progress in computing the function value. They show that, whenever the greedy algorithm queries a variable $x_i$, the amount of progress is at least $\eps / \optworst_0$ times the cost of $x_i$. Since the algorithm makes a progress of $\le \TInf[f]$ in total, the total cost paid by the algorithm is upper bounded by $(\optworst_0 / \eps) \cdot \TInf[f] = \optworst_0 \cdot (\TInf[f] / \eps)$.

However, in the online setting, the analogous claim no longer holds for the \WIPRR{} algorithm. Before observing any input bits, \WIPRR{} might have already invested substantially in all the $n$ variables, making the total investment \emph{much} higher than the progress gained from revealing a single bit. This suggests that the ``progress-to-cost'' ratio cannot be lower bounded as in the offline setting.

Instead, our analysis controls the expected investment in each of the $n$ variables separately. We show that, when \WIPRR{} terminates, the investment in each variable $x_i$ is at most $\optworst_0 / \eps$ times the maximum of $\Inf_i[f_\pi]$, where $\pi$ ranges over all restrictions that \WIPRR{} encounters. In general, this maximum influence \emph{can} be much higher than $\Inf_i[f]$. Fortunately, the trajectory of $\Inf_i[f_\pi]$ throughout the execution forms a martingale bounded between $[0, 1]$, which allows us to upper bound the expectation of $\max_{\pi}\Inf_i[f_\pi]$ by $\Inf_i[f]$ up to a logarithmic factor. Summing over all $i \in [n]$ gives the upper bound of $\optworst_0 \cdot O(\TInf[f] \log(n) / \eps)$.

\paragraph{An Improved Bound for Symmetric Functions.} We prove \Cref{thm:symmetric-informal} by comparing how the optimal zero-error offline algorithm and the online algorithm \WIPRR{} work on a symmetric function. Knowing the costs of the variables, the optimal offline algorithm queries the variables in increasing order of costs. Furthermore, since the algorithm must have a zero error, it stops only if function $f$ reduces to a constant given the observed inputs.

For our online algorithm, since every variable is equally influential with respect to a symmetric function (or a restriction thereof), \WIPRR{} maintains the same investment level across all inputs that have not been revealed. Consequently, the variables also get revealed in increasing order of costs. There are two differences: (1) \WIPRR{} stops as soon as the restricted function become $\eps$-close to a constant; (2) the algorithm pays an additional cost for each variable that is not revealed at the end. \Cref{thm:symmetric-informal} follows from analyzing the stopping times of these two different yet similar processes.

\paragraph{Competing Against Average-Case Benchmarks.} To derive our algorithm that competes against $\optavg_\eps$ in \Cref{thm:iprr-informal}, we take a detour and see why the \WIPRR{} algorithm fails to even compete with $\optavg_0$ without incurring an $\widetilde\Omega(\sqrt{n})$ competitive ratio:

\begin{proposition}[Informal version of \Cref{prop:warmup-iprr-hard-instance}]\label{prop:warmup-iprr-hard-instance-informal}
    There is a function $f\isazofunc$ such that, for any accuracy parameter $\eps \in (0, 1/4)$, $\WIPRR(f, \eps)$ gives an expected cost of
    \[
        \optavg_0 \cdot \TInf[f] \cdot \widetilde{\Omega}(\sqrt{n}).
    \]
\end{proposition}

The hard instance in \Cref{prop:warmup-iprr-hard-instance-informal} is constructed as follows. First, we find a function $g(x)$ such that $\TInf[g] = \Theta(\sqrt{n})$, while \WIPRR{} pays an expected cost of $\optavg_0(g) \cdot \widetilde\Omega(\sqrt{n})$ on $g$.\footnote{Note that $g$ alone does not give the hard instance in \Cref{prop:warmup-iprr-hard-instance-informal}.} Then, we ``dilute'' $g(x)$ by constructing an alternative function $f(x, y)$ such that, over the randomness in $y$: (1) with probability $p \coloneqq 1/\sqrt{n}$, $f(x, y)$ is set to $g(x)$; (2) with probability $1 - p$, $f(x, y) = h(y)$ for some simple function $h$. Furthermore, each additional input bit in $y$ has a cost of zero. This dilution ensures that: (1) $\optavg_0(f) = \optavg_0(g) / \sqrt{n}$, since the offline algorithm for $f(x, y)$ needs to evaluate $g(x)$ only with probability $p = 1/\sqrt{n}$; (2) $\TInf[f] \le p \cdot \TInf[g] + O(1) = O(1)$; (3) The cost of \WIPRR{} on $f$ is a $p$-fraction of its cost on $g$, namely,
\[
    p\cdot \left[\optavg_0(g) \cdot \widetilde\Omega(\sqrt{n})\right]
=   \optavg_0(g) \cdot \widetilde\Omega(1)
=   \optavg_0(f) \cdot \TInf[f] \cdot \widetilde\Omega(\sqrt{n}).
\]

The hard instance above exploits the fact that \WIPRR{} uses a \emph{local} stopping criterion: the algorithm stops only if the current restriction is $\eps$-close to a constant function. Suppose that $\eps = \Omega(1)$. In the rare case that $f(x, y)$ is determined by $g(x)$ (which happens with probability $p = 1/\sqrt{n} \ll \eps$), \WIPRR{} would still pay the cost for evaluating $g(x)$. On the other hand, a more ``global'' algorithm might realize that $f(x, y)$ is $\eps$-close to $h(y)$, so it might as well compute the simpler function $h$ instead.

Indeed, our algorithm for competing against $\optavg_\eps$, denoted by \IPRR{}, simply replaces the stopping condition of \WIPRR{}: The algorithm takes an additional parameter $B$ as input, and stops by outputting the value $\ex{\bx\sim\zo^n}{f_\pi(\bx)}$ rounded to $\{0, 1\}$ whenever $f_\pi$---the current restriction of the function---satisfies the condition\footnote{This can be verified by checking whether $\theta_i$ exceeds $(B/\eps)\cdot\Inf_i[f_\pi]$ for every unrevealed variable $x_i$.}
\[
    \frac{\Inf_i[f_\pi]}{c_i} \le \frac{\eps}{B},~\forall i \in [n].
\]

Via an analysis similar to that of \Cref{thm:warmup-iprr-informal}, we show that $\IPRR(B)$ gives an $O(\eps)$ error as long as $B \ge \optavg_\eps / \eps$, and the resulting expected cost is $B \cdot O((\TInf[f] \log n) / \eps)$. In particular, the first part of \Cref{thm:iprr-informal} follows from setting $B = \optavg_\eps / \eps$, so that the expected cost is $\optavg_\eps \cdot O((\TInf[f] \log n) / \eps^2)$. When $\optavg_\eps$ is unknown, a natural strategy is to guess the value $B = 2^1, 2^2, 2^3, \ldots$ through doubling. For each guess of $B$, we check whether $\IPRR(B)$ is $O(\eps)$-error by sampling $\widetilde O(1/\eps)$ inputs uniformly at random from $\zo^n$, and estimate the empirical error of $\IPRR(B)$. If the error is indeed $O(\eps)$, we run $\IPRR(B)$ on the actual unknown input $x$. This ensures that we stop with high probability as soon as the guess of $B$ reaches $\optavg_\eps / \eps$. This testing procedure introduces an additional $\widetilde O(1/\eps)$ factor in the competitive ratio.

\paragraph{Functions with Small-Depth Decision Trees.} We prove \Cref{thm:follow-the-tree-general-informal} using algorithms that are significantly different from the \IPRR{} strategy. Rather than simultaneously investing in multiple variables in proportion to their influences, we focus on finding an ``efficient'' decision tree representation of the function, so that ``following the tree'' (i.e., querying variables as suggested by the decision tree) \emph{always} leads to a low competitive ratio, regardless of the variable costs.

Formally, we show that, if the given decision tree representation $T$ is \emph{everywhere $\tau$-influential} (i.e., every node in the tree queries a variable with influence $\ge \tau$), simply following tree $T$ gives a $(1/\tau)$-competitive algorithm. More generally, if $T$ has average depth $d$ and is not guaranteed to be everywhere-influential, we apply the pruning lemma of~\cite{BLQT22} to transform $T$ into an everywhere $(\eps / d)$-influential tree $T'$ while introducing an error $\le \eps$. Then, following the pruned tree $T'$ gives a $(d/\eps)$-competitive algorithm. Finally, when $T$ is not given, the learning algorithm of~\cite{BLQT22} allows us to compute an everywhere $\Omega(\eps/d)$-influential decision tree that approximates $f$ up to an error of $\eps$. This algorithm only requires query access to $f$, and runs in time $(d/\eps)^{O(d/\eps)}$. The knowledge of $d$ can be further removed by a standard doubling trick.

\paragraph{Competitive Ratio Lower Bounds.} 

We first prove the following lower bound against zero-error algorithms computing the $\AND$ function: 

\begin{proposition}[Informal version of \Cref{prop:and-lower-bound}] \label{prop:and-lower-bound-informal}
    Let $f\isazofunc$ be the $\AND$ function on $n$ variables. 
    Then every zero-error online algorithm computing $f$ has expected cost $\optavg_0 \cdot \Omega(n)$. 
\end{proposition}

We prove \Cref{prop:and-lower-bound-informal} by choosing the costs of the $n$ variables in the $\AND$ function as a uniformly random permutation of $[n] = \{1, 2, \ldots, n\}$. The optimal offline algorithm queries the variables in ascending order of costs and incurs an $O(1)$ cost in expectation. The online algorithm, however, does not know which variables are the cheapest. Intuitively, the best strategy is to invest in the $n$ variables in a round robin fashion, which incurs an $\Omega(n)$ cost even before seeing a single input bit.

For \Cref{thm:tribes-lower-bound-informal}, recall that an $n$-variable $\Tribes$ function is the $\OR$ of $n / w = \Theta(n/\log n)$ disjoint tribes, each of which is the $\AND$ of $w = \Theta(\log n)$ variables. We choose the costs of the variables in each tribe as a random permutation of $[w]$, so that the $\Tribes$ instance consists of $n/w$ independent copies of the $\AND$ instance in \Cref{prop:and-lower-bound-informal}. To lower bound the cost of the online algorithm, we show that any low-error algorithm for $\Tribes$ must observe at least one variable from an $\Omega(1)$-fraction of the tribes. Then, by planting a $w$-variable $\AND$ instance as one of the $n/w$ tribes, we obtain an algorithm for the $\AND$ function that reveals at least one of the variables with probability $\Omega(1)$. By a strengthening of \Cref{prop:and-lower-bound-informal}, such an algorithm for $\AND$ must have an $\Omega(w)$ competitive ratio, which then translates into an $\Omega(\log n)$ competitive ratio for the $\Tribes$ instance.
\section{Preliminaries}
\label{sec:prelims}

We use boldfaced letters (e.g.~$\bx \sim \zo^n$) to denote random variables.
Unless explicitly stated otherwise, all probabilities and expectations will be with respect to the uniform distribution. 
All logarithms will be with respect to base-$2$, unless explicitly stated otherwise. 
We write $[n] := \{1,\dots,n\}$ and $\{e_i\}_{i=1}^n$ for the collection of standard basis vectors in $\R^n$,~i.e. $e_i = (0, \dots, 0, 1, 0, \dots, 0).$

\subsection{Boolean Functions}
\label{subsec:influences-OSSS}

Our notation and terminology follow~\cite{odonnell-book}. Given two Boolean functions $f,g:\zo^n\to\zo$, we define 
\begin{equation} 
\label{eq:hamming-dist}
	\dist(f,g) := \Prx_{\bx\sim\zo^n}[f(\bx) \neq g(\bx)]\,.
\end{equation}
Given $f:\zo^n \to \zo$, we will write
\[
    \bias(f) := \min\{\dist(f, 0), \dist(f, 1)\}\,.
\]

We recall the notion a variable's \emph{influence} on a Boolean function (see Chapter~2 of~\cite{odonnell-book} for further background and information) defined earlier in~\Cref{sec:intro}:

\infDef*

Given a query algorithm $\calA$, it will be convenient to write 
\[
    \error_{f}(\calA) := \dist(f,\calA)
\]
where we identify $\calA:\zo^n\to\zo$ with the decision tree induced by its computation. 
We will sometimes say that ``$\calA$ is an $\eps$-error query algorithm for $f$'' if $\error_f(\calA) \leq \eps$. 
For every index $i \in [n]$, we write 
\[
    \delta_i(\calA) := \Prx_{\bx\sim\zo^n}\sbra{\calA~\text{queries}~\bx_i}\,. 
\]
Note that the expected number of queries to compute $f$ via query algorithm $\calA$ can be written in terms of the $\delta_i(\calA)$'s as follows:
\[
 \sumi \delta_i(\calA).
\]
Our algorithms in~\Cref{sec:iprr} will rely on the well-known ``OSSS inequality'' from the analysis of Boolean functions~\cite{OSSS:05} (see also Chapter~8 of~\cite{odonnell-book}). 
The following variant of the OSSS inequality is due to Jain and Zhang~\cite{JainZhang2011}:

\begin{theorem}[OSSS inequality, \cite{JainZhang2011} version]
\label{thm:osss-JZ-version}
	For all functions $f:\zo^n\to\zo$ and query algorithms $\calA$, we have 
    {
	\[
		\bias(f) - \error_f(\calA) \leq \sumi \delta_i(\calA)\cdot \Inf_i[f]. 
	\]
    }
\end{theorem}

Note that \Cref{thm:osss-JZ-version} is a refinement of the classical edge-isoperimetric or Poincar\'{e} inequality (see Chapter~2 of \cite{odonnell-book}) over the Boolean hypercube that takes the query complexity (or alternatively, decision-tree structure) of $f$ into account. 
This inequality underlies the analysis of both the offline strategy of~\cite{blanc2021query} as well as our online algorithms in~\Cref{sec:iprr}. 

\subsection{Problem Setup}
\label{subsec:problem-setup}

We begin by recalling the setup of Blanc, Lange, and Tan~\cite{blanc2021query}, which is itself a variant of the problem formulation considered by Deshpande, Hellerstein, and Kletenik~\cite{deshpande2014approximation}.
In this setting, the algorithm is given a function $f:\zo^n\to\zo$ and a cost vector $c \in \R_{\geq 0}^n$. 
Its objective is to design a query algorithm to compute $f$ on an unknown (but fixed) input $x \in \zo^n$, where it incurs a cost of $c_i$ to reveal the $i^\text{th}$ bit of $x$. 
We will sometimes refer to this setup, where the cost vector $c$ is known to the algorithms, as the ``offline'' setting.

In our setting, the key difference is that $c_i$ is \emph{unknown} to the algorithm. To distinguish this setup from the previous one, we will refer to it as the ``online'' setting.
\begin{definition}[Online Priced Query Model]
\label{def:model}
    The algorithm is given a Boolean function $f: \zo^n \to \zo$, while the input $x \in \zo^n$ and the cost vector $c \in \R_{\ge 0}^n$ are unknown. The algorithm maintains a \emph{cumulative investment vector} $\theta \in \mathbb{R}_{\geq 0}^n$, where $\theta_i$ represents the algorithm’s total investment towards revealing the variable $x_i$. At each step, the algorithm selects a coordinate $i \in [n]$ to invest in, incrementing $\theta_i$ by $\beta$ (which can be viewed as the algorithm’s minimal budget or resolution). A bit $x_i$ is revealed once the investment $\theta_i$ reaches or exceeds the corresponding cost $c_i$. 
    The \emph{total cost} incurred by the algorithm is defined as the value of $\|\theta\|_1$ when it halts.
\end{definition}

\begin{remark}
\label{remark:f-access}
    Throughout, we assume that our algorithms are given an explicit representation of the function $f\isazofunc$. 
    They can be modified in a natural fashion to work with black-box or query access to $f$. 
    For example, the \WIPRR{} and \IPRR{} algorithms in~\Cref{sec:iprr} require knowledge of coordinate influences, which can be estimated via random sampling; the failure probability can then be controlled by a union bound.
\end{remark}

\begin{remark}
   Our upper bounds will exhibit a dependence on a ``unit investment'' or ``resolution''  parameter $\beta$, which represents the algorithm's unit investment step size.
   However, note that $\beta$ is a free parameter that can be chosen to be arbitrarily small. 
\end{remark}

The following notation will be helpful: 

\begin{definition}
    Let $\calA$ be an $\eps$-error (offline or online) algorithm for computing $f:\zo^n\to\zo$, and let $c \in \R_{\geq 0}^n$ be an associated cost vector. 
    We will write $\cost^f_c(\calA, x)$ for the total cost incurred by $\calA$ on input $x \in \zo^n$. 
    We define
    \[
        \costavg^f_c(\calA) := \Ex_{\bx\sim\zo^n}\sbra{\cost_c(\calA, \bx)}
        \qquad\text{and}\qquad 
        \costworst^f_c(\calA) := \max_{x\in\zo^n} \cost_c(\calA, x)
    \]
    to be the \emph{expected cost} and \emph{worst-case cost} of $\calA$ respectively. 
    When the function $f$ is clear from context, we will omit dependence on it (writing, for e.g., $\costavg_c(\calA)$ instead of $\costavg_c^f(S)$) for notational simplicity.
\end{definition}

We will benchmark the performance of our algorithms against the cost of optimal query algorithms, which may be either offline or online:

\begin{definition}
    \label{def:benchmark}
    Given a Boolean function $f:\zo^n\to\zo$, a cost vector $c\in\R^n_{\geq 0}$, and an error parameter $\eps \in [0,0.5]$, we define 
    \[
        \optavg_\eps(f,c) := \inf \costavg_c^f(\calA)
        \qquad\text{and}\qquad 
        \optworst_\eps(f,c) := \inf\costworst_c^f(\calA)\,,
    \]
    where both infima are taken over all $\eps$-error query algorithms $\calA$ (either offline or online) for $f$. 
    We will frequently write $\optavg_\eps = \optavg_\eps(f,c)$ and $\optworst_\eps = \optworst_\eps(f,c)$ for simplicity whenever $f$ and $c$ are clear from context.
\end{definition}

It is readily verified that 
\begin{equation} 
\label{eq:opt-relationships}
    \optavg_\eps \leq \optavg_0 \leq \optworst_0  
    \qquad\text{and}\qquad 
    \optavg_\eps \leq \optworst_\eps \leq \optworst_0\,. 
\end{equation}
Thus, benchmarking against $\optavg_\eps$ is the strongest guarantee we can hope for.
\usetikzlibrary{decorations.pathreplacing}

\section{The Influence-Proportional Round Robin Algorithm}
\label{sec:iprr}

The main result of this section is an $O(\eps)$-error online algorithm that is competitive with the optimal $\eps$-error \emph{offline} algorithm, up to an additional factor of $\TInf[f]$ and some logarithmic terms: 

\begin{theorem} \label{thm:iprr-main}
    Let $f\isazofunc$ be a Boolean function and $c \in \R_{\geq 0}$ be an unknown cost vector. 
    For $\eps \in (0, 0.5]$ and $\beta > 0$, there exists an online algorithm with unit investment $\beta$, \fullAlg~(\Cref{alg:full-alg}) which:
    \begin{itemize}
    	\item Computes $f$ to $O(\eps)$-error; and
    	\item Assuming $\optavg_\eps = \Omega(1)$, satisfies
    		\begin{equation} \label{eq:iprr-main-goal}
                \costavg_c(\fullAlg) \leq\beta n + \optavg_\eps \cdot O\pbra{\frac{1}{\eps^3}\log\frac{\log\optavg_\eps}{\eps}}\cdot\sum_{i=1}^{n}\Inf_i[f]\pbra{1 + \ln\frac{1}{\Inf_i[f]}}\,.
   			\end{equation}
    \end{itemize}
\end{theorem} 

\begin{remark}
    Since the function $x \mapsto x(1+\ln x^{-1})$ is concave, it follows from \Cref{eq:iprr-main-goal} and Jensen's inequality that 
    \[
        \costavg_c(\fullAlg) \leq \beta n + \optavg_\eps\cdot O\pbra{\frac{\TInf[f]\log n}{\eps^3}\cdot\log\frac{\log\optavg_\eps}{\eps}}\,.
    \]
\end{remark}

\subsection{Warmup}
\label{subsec:warmup-iprr}

Before proving~\Cref{thm:iprr-main}, we first establish a weaker result which illustrates the key ideas behind the proof of~\Cref{thm:iprr-main}. 
In particular, some parts of the proof of~\Cref{thm:iprr-main} will rely on calculations from the proof of~\Cref{thm:warmup-iprr} below.

\begin{theorem}[Formal version of \Cref{thm:warmup-iprr-informal}]\label{thm:warmup-iprr}
    Let $f:\zo^n\to\zo$ be a Boolean function and $c \in \R^n_{\geq 0}$ be a cost vector.
    For $\eps \in (0,0.5]$ and $\beta > 0$, there exists an $\eps$-error online algorithm  \WIPRR~(\Cref{alg:warmup-iprr}) with unit investment $\beta$ such that
    \begin{equation} \label{eq:iprr-warmup-goal}
        \costavg_c(\WIPRR) \leq \beta n + \frac{\optworst_0}{\eps}\cdot \sumi \Inf_i[f]\pbra{1 + \ln\frac{1}{\Inf_i[f]}}\,.
    \end{equation}
\end{theorem}

\Cref{thm:warmup-iprr} gives an $\eps$-error algorithm whose expected cost is upper bounded by the (offline) \emph{zero-error} algorithm with minimal worst-case cost---a guarantee that is the weakest among those we consider (cf.~\Cref{eq:opt-relationships}). 
(In contrast,~\Cref{thm:iprr-main} is competitive against the optimal $\eps$-error algorithm with minimal average-case cost, which is the strongest guarantee we consider.)

\begin{algorithm}
    \caption{The $\WIPRR$ Algorithm}
    \label{alg:warmup-iprr}

    \vspace{1em}
    \textbf{Input:} Succinct representation of $f:\zo^n\to\zo$, error parameter $\eps \in (0, 0.5]$ \\[0.25em]
    \textbf{Output:} A bit $b \in \zo$

    \ 

    $\WIPRR(f, \eps)$:
    \begin{enumerate}
        \item Initialize $\theta \leftarrow 0^n$ and $\pi \leftarrow \emptyset$. 
        \item While $\bias(f_\pi) > \eps$:
        \begin{enumerate}
            \item Let $i^\ast \in [n]$ be the index such that 
            \[
                i^\ast = \arg\max_i \frac{\Inf_i[f_{\pi}]}{\theta_i}\,.
            \]
            \item Spend cost $\beta$ towards $x_{i^\ast}$ and update $\theta \leftarrow \theta + \beta e_{i^\ast}$.
            \item If $x_{i^\ast}$ is revealed to be $b_{i^\ast} \in \zo$, then update $\pi \leftarrow \pi \cup \{x_{i^\ast} \mapsto b_{i^\ast}\}$.
        \end{enumerate}
        \item Output $\displaystyle\mathbf{1}\bigg\{\ex{\bx\sim\zo^n}{f_\pi(\bx)} \ge 1/2\bigg\}$.
    \end{enumerate}   
\end{algorithm}

Our proof of~\Cref{thm:warmup-iprr} will rely on the following consequence of the well-known OSSS inequality (\Cref{thm:osss-JZ-version}) obtained by Blanc, Lange, and Tan~\cite{blanc2021query}: 

\begin{lemma}[Lemma~3.1 of~\cite{blanc2021query}] \label{lemma:OSSS-consequence}
    For all Boolean functions $f:\zo^n\to\zo$, cost vectors $c \in \R_{\geq 0}^n$, and query algorithms $\calA$, we have 
    \[
        \max_{i\in[n]} \frac{\Inf_i[f]}{c_i} \geq \frac{\bias(f) - \error_f(\calA)}{\costavg^f_c(\calA)}\,.
    \]
\end{lemma}

We will also require Doob's martingale inequality; recall that a discrete-time stochastic process $(\bX_1, \dots, \bX_T)$ is a \emph{martingale} if for all $t \in [T]$, we have 
\[
    \Ex\sbra{|\bX_t|} < \infty 
    \qquad\text{and}\qquad  
    \Ex\sbra{\bX_{t+1} \mid \bX_1, \dots, \bX_t} = \bX_t\,.
\]

\begin{lemma}[Doob's inequality, Section~14.6 of~\cite{williams1991probability}] \label{lemma:doob}
    Given a martingale $(\bX_1, \dots, \bX_T)$, for all $C > 0$ we have 
    \[
        \Pr\sbra{\max_{t \in [T]} \bX_t \geq C} \leq \frac{\Ex\sbra{\max\cbra{\bX_T, 0}}}{C}\,.
    \]
\end{lemma}

We are now ready to prove~\Cref{thm:warmup-iprr}:

\begin{proofof}{\Cref{thm:warmup-iprr}}
    Consider the execution of the algorithm \WIPRR~on a fixed input $x \in \zo^n$ with unknown cost vector $c$. Note that $x$ induces a sequence of restrictions $\pi$ corresponding to the variables revealed by \WIPRR; call this collection of restrictions $\Pi(x)$.
    For each index $i \in [n]$, let $\pi^{(i)}$ denote the restriction $\pi$  \emph{right before} 
    $\theta_i$ (which is the investment in variable $x_i$) gets incremented for the last time. 
    In the rest of the proof, vector $\theta$ denotes the investments in the $n$ variables \emph{at this moment} (before $\theta_i$ gets incremented for the last time). In contrast, the \emph{final} investment in variable $x_i$ will be denoted by $\theta_i^\ast$, which is equal to $\theta_i + \beta$.

    Let $\calA^{\circ}$ denote the zero-error offline algorithm for $f$ with minimal worst-case cost, i.e.,\footnote{Technically, $\optworst_0$ is defined as an infimum and might not be obtained by any algorithm, though the rest of the proof would still go through by considering a sequence of zero-error algorithms with worst-case costs approaching $\optworst_0$.}
    \[
        \error_f(\calA^{\circ}) = 0 \qquad\text{and}\qquad \costworst^f_c(\calA^{\circ}) = \optworst_0(f,c)\,.
    \]
    It follows from Step~2(a) of~\Cref{alg:warmup-iprr} that
	\[
		\frac{\Inf_i[f_{\pi^{(i)}}]}{\theta_i^\ast - \beta} = \frac{\Inf_i[f_{\pi^{(i)}}]}{\theta_i} = \max_j \frac{\Inf_j[f_{\pi^{(i)}}]}{\theta_j} \geq \max \frac{\Inf_j[f_{\pi^{(i)}}]}{c_j}\,.
	\]
    Let $\calA^{\circ}_{\pi^{(i)}}$ denote the algorithm obtained from $\calA^{\circ}$ by enforcing the restriction $\pi^{(i)}$. In other words, $\calA^{\circ}_{\pi^{(i)}}$ follows $\calA^{\circ}$ and, whenever a variable $x_j$ is queried by $\calA^{\circ}$ while $x_j \mapsto b_j$ is among restriction $\pi^{(i)}$, $\calA^{\circ}_{\pi^{(i)}}$ proceeds by forwarding $x_j = b_j$ to $\calA^{\circ}$. By construction, we have
    \[
        \error_{f_{\pi^{(i)}}}(\calA^{\circ}_{\pi^{(i)}}) = 0
    \qquad\text{and}\qquad 
        \costworst_c^{f_{\pi^{(i)}}}(\calA^{\circ}_{\pi^{(i)}})
    \le \costworst_c^f(\calA^{\circ})
    =   \optworst_0(f, c).
    \]

    Then, applying~\Cref{lemma:OSSS-consequence} to $f_{\pi^{(i)}}$ and $\calA^{\circ}_{\pi^{(i)}}$ gives
	\begin{equation} \label{eq:warmup-osss-1}
		\frac{\Inf_i[f_{\pi^{(i)}}]}{\theta_i^\ast - \beta} 
    \geq \max \frac{\Inf_j[f_{\pi^{(i)}}]}{c_j}
    \geq \frac{\bias(f_{\pi^{(i)}})}{\costavg_c^{f_{\pi^{(i)}}}(\calA^{\circ}_{\pi^{(i)}})} \geq \frac{\eps}{\optworst_0}\,,
	\end{equation}
    where the last inequality relies on:
    \begin{itemize}
        \item The observation 
		\[
			\costavg_c^{f_{\pi^{(i)}}}(\calA^{\circ}_{\pi^{(i)}})
			\leq 
			\costworst_c^{f_{\pi^{(i)}}}(\calA^{\circ}_{\pi^{(i)}}) 
			\le 
			\optworst_0(f,c)\,;
		\]
        as well as  
        \item The fact that $\bias(f_{\pi^{(i)}}) > \eps$ by design of~\Cref{alg:warmup-iprr}. 
    \end{itemize}
    We can rewrite~\Cref{eq:warmup-osss-1} as follows:
    \begin{equation}\label{eq:warmup-osss-1-rewritten}
        \theta_i^\ast \leq \beta + \frac{\optworst_0}{\eps} \cdot \Inf_i[f_{\pi^{(i)}}] \leq \beta + \frac{\optworst_0}{\eps} \cdot \max_{\pi \in \Pi(x)} \Inf_i[f_{\pi}]\,.
    \end{equation}
    Recall that $\Pi(x)$ is the collection of all restrictions encountered by \WIPRR~on input $x \in \zo^n$. The second step above follows from the observation that $\pi^{(i)} \in \Pi(x)$.
 
    Note that because $\Inf_i[f_{\pi^{(i)}}] \in [0,1]$, \Cref{eq:warmup-osss-1-rewritten} immediately implies that 
    \[
        \costworst_c(\WIPRR) \leq \beta n + \frac{\optworst_0}{\eps}n\,.    
    \]
    We will obtain our improved bound by considering the expected cost of \WIPRR. Indeed, it follows from~\Cref{eq:warmup-osss-1-rewritten} and the definition of $\costavg_c(\cdot)$ that  
    \begin{align*}
        \costavg_c(\WIPRR) 
        &\leq \beta n + \frac{\optworst_0}{\eps} \sumi \Ex_{\bx\sim\zo^n}\sbra{\max_{\pi \in \Pi(\bx)} \Inf_i[f_{\pi}]}.
    \end{align*}
    So in order to complete the proof, it suffices to show that 
    \begin{equation} \label{eq:warmup-final-goal}
        \Ex_{}\sbra{\max_{\pi \in \Pi(\bx)} \Inf_i[f_{\pi}]} 
        \leq \Inf_i[f]\cdot\pbra{1 + \ln\frac{1}{\Inf_i[f]}}. 
    \end{equation} 
    The rest of the argument will establish~\Cref{eq:warmup-final-goal}. 
    
    For each $t \in \{0, 1, 2, \ldots, n\}$, let $\bpi_t$ denote the restriction encountered by \WIPRR{} when it runs on a uniformly random input $\bx\sim\zo^n$ and has revealed the value of exactly $t$ variables. If the algorithm reveals $t' < t$ variables before halting, we define $\bpi_t = \bpi_{t'}$ instead.
	Then, define the random variable $\bX^{(i)}_t$ as follows:
	\begin{itemize}
		\item If $\bx_i$ is not among the first $t$ revealed variables, then $\bX^{(i)}_t = \Inf_{i}[f_{\bpi_t}]$. 
		\item If $\bx_i$ is the $t^\ast$-th revealed variable for some $t^\ast \le t$, then $\bX^{(i)}_t = \bX^{(i)}_{t^\ast - 1}$. 
		In other words, the value of $\bX^{(i)}_{t}$ is frozen to its last value before $\bx_i$ is revealed.
	\end{itemize}	
    It is readily verified that $(\bX^{(i)}_0, \bX^{(i)}_1, \dots , \bX^{(i)}_n)$ forms a martingale with the following properties:
    \begin{itemize}
        \item $\bX^{(i)}_0 = \Inf_i[f]$ with probability $1$.
        \item $\bX^{(i)}_t \in [0,1]$ for all $t \in \{0, \dots, n\}$. 
    \end{itemize}
    Furthermore, we note that $\Pi(x) = \{\bpi_0, \bpi_1, \ldots, \bpi_n\}$, so it holds that $\max_{\pi \in \Pi(x)}\Inf_i[f_{\pi}] = \max_{0 \le t \le n}\bX^{(i)}_t$.
    Let $r^\ast$ be a parameter that we will set shortly. 
    We then have 
    \begin{align}
        \text{L.H.S. of}~\eqref{eq:warmup-final-goal} 
        = \Ex\sbra{\max_{0 \leq t \leq n} \bX^{(i)}_t } 
        &= \int_{r = 0}^\infty \Prx_{}\sbra{\max_{0 \leq t \leq n} \bX^{(i)}_t  \geq r }\,\rmd r \label{eq:sumiao}\\[0.5em] 
        &= \int_{r = 0}^1 \Prx_{}\sbra{\max_{0 \leq t \leq n} \bX^{(i)}_t  \geq r }\,\rmd r \tag{Since $\bX^{(i)}_t \in [0,1]$ for all $t$}\\[0.5em] 
        &= \int_{r=0}^{r^\ast} \Prx_{}\sbra{\max_{0 \leq t \leq n} \bX^{(i)}_t  \geq r }\,\rmd r + \int_{r = r^\ast}^1 \Prx_{}\sbra{\max_{0 \leq t \leq n} \bX^{(i)}_t  \geq r }\,\rmd r \nonumber \\[0.5em]
        &\leq r^\ast + \int_{r = r^\ast}^1  \Prx_{}\sbra{\max_{0 \leq t \leq n} \bX^{(i)}_t  \geq r }\,\rmd r \nonumber \\[0.5em] 
        &\leq r^\ast  + \int_{r = r^\ast}^1 \frac{\Ex_{}\sbra{\bX_n^{(i)}}}{r} \, \rmd r \tag{Doob's inequality,~\Cref{lemma:doob}} \\[0.5em]
        &= r^\ast  + \int_{r = r^\ast}^1 \frac{\Inf_i[f]}{r} \, \rmd r \tag{Since $\bX_t^{(i)}$ is a martingale} \\[0.5em]
        &= r^\ast + \Inf_i[f]\cdot\ln\pbra{\frac{1}{r^\ast}}\,. \label{eq:flour}
    \end{align}
    The final quantity above is minimized for $r^\ast = \Inf_i[f]$, and so 
    \[
        \Ex_{}\sbra{\max_{\pi \in \Pi(\bx)} \Inf_i[f_{\pi}]} \leq \Inf_i[f] \cdot \pbra{1 + \ln\frac{1}{\Inf_i[f]}},
    \]
    establishing~\Cref{eq:warmup-final-goal} and in turn completing the proof.
\end{proofof}

\subsection{A Hard Instance for \WIPRR{}}
\label{subsec:iprr-tightness}

In this section, we will consider an instance of the online query problem which shows that the analysis of~\Cref{alg:warmup-iprr} from the previous section is tight. 
In particular, this ``hard instance'' suggests a natural modification to~\Cref{alg:warmup-iprr} which will allow us to establish~\Cref{thm:iprr-main}. 
Throughout this section, we will suppress dependence on the resolution or unit-cost parameter $\beta$ and will assume that $\beta$ is $n^{-c}$ for some sufficiently large constant $c$. 
This is without loss of generality since for any choice of $\beta$, the cost vector $c$ can be scaled appropriately. 

\begin{proposition}[Formal version of \Cref{prop:warmup-iprr-hard-instance-informal}]
\label{prop:warmup-iprr-hard-instance}
    For every integer $k \ge 1$ and $n = 2^{2k} + 2k$, there is an $n$-variable function $h$ and a cost vector $c \in \R_{\geq 0}^n$ such that, for any accuracy parameter $\eps \in (0, 1/4)$, $\WIPRR(h, \eps)$ gives an expected cost of
    \[
        \costavg^h_c(\WIPRR) = \optavg_0(h, c) \cdot \Omega\left(\TInf[f]\cdot \frac{\sqrt{n}}{\log n}\right).
    \]
\end{proposition}

We will require a slight modification of the well-known ``address function'' from Boolean function complexity: 

\begin{definition} \label{def:address-func}
    For $n = 2^{2k}$ for some $k \in \N$, consider the function $f: \zo^{\log \sqrt{n}} \times \zo^{n} \to \zo$ defined as 
    \[
        f(x, y) := \bigoplus_{i=1}^{\sqrt{n}} y_{x, i} 
    \]
    where we identify $x \in \zo^{\log \sqrt{n}}$ with the binary representation of an index $j \in [\sqrt{n}]$ and view $y \in \zo^n$ as a $\sqrt{n}\times\sqrt{n}$ Boolean matrix. (Here, $x\oplus y = (x + y) \bmod 2$ is the XOR operation.)
\end{definition}

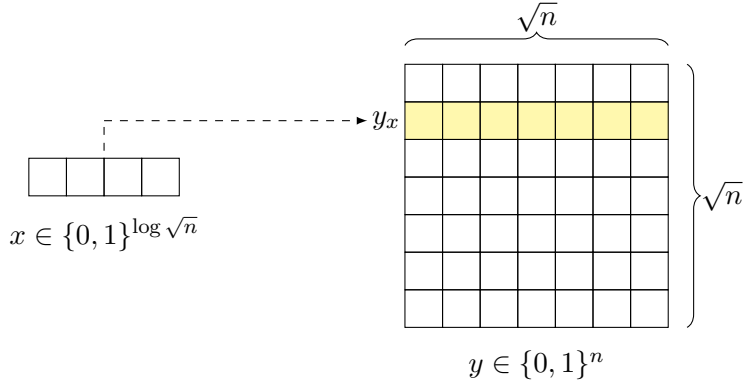
\begin{figure}
    \centering
    \begin{tikzpicture}
    
    	\fill[color=yellow!40] (5,-0.5) rectangle ++(3.5,0.5);
    	\node (yx) at (4.75,-0.25) {$y_x$};
    	\draw[dashed] (1,-0.65) -- (1,-0.25);
    	\draw[-latex,dashed] (1,-0.25) -- (4.5,-0.25); 
    
    	\foreach \x in {0,...,3} {
    		\draw (\x/2,-1.25) rectangle ++(0.5,0.5);
    	}
    	
    	\node (x) at (1,-1.75) {$x \in \zo^{\log \sqrt{n}}$};
    	
    	\foreach \x in {0,1,...,6} {
    		\foreach \y in {0,1,...,6} {
    			\draw (5+\x/2,-\y/2) rectangle ++(0.5,0.5);
    		}
    	}
    	
    	\node (y) at (6.75,-3.5) {$y \in \zo^n$};
    	\draw [decorate,decoration={brace,amplitude=5pt,mirror}]
  (8.5,0.75) -- (5,0.75) node[midway,yshift=1em]{$\sqrt{n}$};\draw [decorate,decoration={brace,amplitude=5pt}]
  (8.75,0.5) -- (8.75,-3) node[midway,xshift=1.2em]{$\sqrt{n}$};
    	
	\end{tikzpicture}
    \caption{The function $f: \zo^{\log \sqrt{n}} \times \zo^{n} \to \zo$ from \Cref{def:address-func}.}
    \label{fig:address-func}
\end{figure}

It will be convenient to refer to the variables of $x\in\zo^{\log\sqrt{n}}$ as \emph{control variables}, and to the variables of $y \in \zo^n$ as \emph{action variables}. 
(See~\Cref{fig:address-func} for an illustration of~\Cref{def:address-func}.)
It is readily verified that
\[
	\Inf_i[f] = 
	\begin{cases}
		0.5 & i~\text{indexes a control variable} \\ 
		\frac{1}{\sqrt{n}} & i~\text{indexes an action variable}
	\end{cases}\,
\]
and consequently $\TInf[f] = \Theta(\sqrt{n}).$

Consider a cost vector $c \in \R^{\log\sqrt{n} + n}$ where
\[
    c_i = \Theta_{\beta}\pbra{\Inf_i[f]}\,.
\]
It is readily seen that $\optavg_0(f, c) = \Theta_{\beta}(\log n)$ by considering the algorithm which queries the control variables at cost $\Theta_{\beta}(\log n)$ to obtain $x \in \zo^{\log\sqrt{n}}$ and then computes $\bigoplus_i y_{x,i}$ at cost $\Theta_{\beta}(1)$. 
On the other hand, since $\WIPRR$ (\Cref{alg:warmup-iprr}) invests in proportion to the coordinate influences (which are all equal for the action variables), we have that  
\[
	\costavg_c(\WIPRR) = \Theta_{\beta}(\sqrt{n})
\]
since each of the $n$ action variables has cost $\Theta_\beta(n^{-1/2})$. 
In particular, this example establishes that 
\begin{equation} \label{eq:brookline}
	\costavg_c(\WIPRR) = \Theta_{\beta}(\sqrt{n}) \geq  \Omega\left(\optavg_0 \cdot \frac{\TInf[f]}{\log n}\right)\,. 
\end{equation}

Next, we show how to bootstrap this example to give a query instance of a Boolean function $h:\zo^n \to \zo$ where 
\begin{equation} \label{eq:wiprr-lb-goal}
	\costavg_c(\WIPRR) = \optavg_0(h, c) \cdot \Omega\left(\TInf[h]\cdot \frac{\sqrt{n}}{\log n}\right)\,.
\end{equation}
We first describe the instance which establishes~\Cref{eq:wiprr-lb-goal}. 
Let $\ell \geq 1$
be a parameter to be set shortly, and let $g:\zo^\ell\to\zo$ be a \emph{decision list}~\cite{riv87} (or equivalently, rank-$1$ decision tree)
such that: (i) the decision list computing $g(z)$ queries variables $z_1, z_2, \ldots, z_{\ell}$ in order; (ii) the labels of the decision list's rules alternate (i.e., the first rule is ``If $z_1 = 1$, output $1$; else if $z_2 = 1$, output $0$; \dots''); (iii) the final rule is ``If $z_\ell = 1$, output $\ell \bmod 2$; else, output $1$.'' 

Then, we define the function 
\[
	h:\zo^\ell \times \zo^{\log\sqrt{n}} \times \zo^{n} \to \zo
\]
by modifying $g$ as described in~\Cref{fig:DL}: $h(z, x, y)$ agrees with $g(z)$ as long as $z \ne 0^{\ell}$; when $z = 0^{\ell}$, $h(z, x, y)$ is defined as $f(x, y)$ from \Cref{def:address-func} instead.
We will refer to the variables in $\zo^\ell$ as the \emph{list variables}. 
 
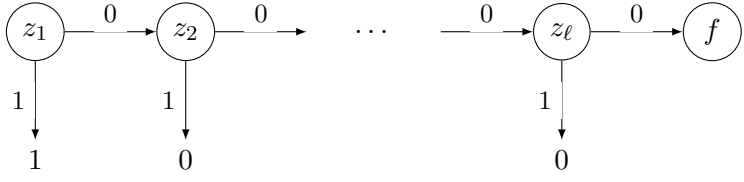
\begin{figure}
    \centering
    \begin{tikzpicture}
        \node[draw,circle] (z1) at (0,0) {$z_1$};	
        \node[draw,circle] (z2) at (2,0) {$z_2$};	
        \node[draw,circle,color=white] (zp) at (4,0) {\textcolor{white}{$z_2$}};	
        \node[draw,circle,color=white] (zp2) at (5,0) {\textcolor{white}{$z_2$}};	
        \node[draw,circle] (z3) at (7,0) {$z_\ell$};	
        \node[draw,circle] (z4) at (9,0) {$f$};	
        \node (dots) at (4.5,0) {$\dots$};

        \draw[-latex] (z1) -- (z2) node [midway, above, fill=white,] {\small $0$};
        \draw[-latex] (z2) -- (zp) node [midway, above, fill=white] {\small $0$};
        \draw[-latex] (zp2) -- (z3) node [midway, above, fill=white] {\small $0$};
        \draw[-latex] (z3) -- (z4) node [midway, above, fill=white] {\small $0$};

        \node[] (b1) at (0,-1.7) {$1$};
        \node[] (b2) at (2,-1.7) {$0$};
        \node[] (b3) at (7,-1.7) {$0$};

        \draw[-latex] (z1) -- (b1) node [midway, left, fill=white] {\small $1$};
        \draw[-latex] (z2) -- (b2) node [midway, left, fill=white] {\small $1$};
        \draw[-latex] (z3) -- (b3) node [midway, left, fill=white] {\small $1$};
        
    \end{tikzpicture}

    \
    
    \caption{The function $h: \zo^{\ell}\times \zo^{\log \sqrt{n}} \times \zo^{n} \to \zo$ where $f$ is as in~\Cref{def:address-func}.}
    \label{fig:DL}
\end{figure}

We thus have
\[
    \Inf_i[h] = 
    \begin{cases}
        \Theta(2^{-i}) & i~\text{is a list variable},\\
        2^{-(\ell+1)} &i~\text{is a control variable}, \\
        \frac{2^{-\ell}}{\sqrt{n}} &i~\text{is an action variable}.
    \end{cases}
\]
In particular, taking $\ell := \log\sqrt{n}$, we get 
\begin{equation}
\label{eq:h-tinf}
    \TInf[h] = \pbra{\sum_{i=1}^{\ell} \Theta(2^{-i})} + 2^{-\ell}\cdot\TInf[f] = \Theta(1)\,. 
\end{equation}

Furthermore, the alternating labels in $g(z)$ ensures that
\begin{equation} 
\label{eq:decision-list-stays-balanced}
    \Ex_{(\bz, \bx, \by)\sim\zo^{\ell + \log\sqrt{n} + n}}\sbra{h_\pi(\bz, \bx, \by)} \in \sbra{\frac{1}{4}, \frac{3}{4}}
\end{equation}
holds for any restriction $\pi$ of form $\{z_1 \mapsto 0, z_2 \mapsto 0, \ldots, z_{\ell'} \mapsto 0\}$ where $\ell' \in \{0, 1, 2, \ldots, \ell\}$.

Next, consider an instance of the query problem where the costs of the list variables are $0$, the control variables are $\Theta_{\beta}(1)$, and the cost of each action variable is $\Theta_\beta(n^{-1/2})$. 
It follows that
\[
    \optavg_0(h,c) = 2^{-\ell}\cdot\optavg_0(f,c)\,,
\]
and similarly
\[
    \costavg^h_c(\WIPRR) = 2^{-\ell}\cdot\costavg^f_c(\WIPRR)\,.
\]
In particular, by \Cref{eq:decision-list-stays-balanced} and the assumption that $\eps < 1/4$ in \Cref{prop:warmup-iprr-hard-instance}, $\WIPRR(h, \eps)$ must query all of the $\ell$ list variables when they are all zeros. It then follows from~\Cref{eq:brookline,eq:h-tinf} that 
\[
    \costavg^h_c(\WIPRR) \geq \optavg_0(h)\cdot \Omega\pbra{\TInf[h]\cdot \frac{\sqrt{n}}{\log n}}\,,
\]
establishing~\Cref{eq:wiprr-lb-goal}. 

\subsection{Proof of \Cref{thm:iprr-main}}
\label{subsec:iprr-ub}

It follows from \Cref{prop:warmup-iprr-hard-instance} that we cannot hope for \WIPRR~to be competitive against $\optavg_0$ without losing a $\poly(n)$-factor in the competitive ratio.  
In this section, we describe a modification of \WIPRR{} which gives an algorithm that is competitive against $\optavg_\eps$, establishing~\Cref{thm:iprr-main}. 
We will in fact obtain~\Cref{thm:iprr-main} as an immediate consequence of the following (seemingly) weaker guarantee: 

\begin{proposition} 
\label{prop:markov}
    Let $f\isazofunc$ be a Boolean function, $c \in \R^n_{\geq 0}$ be a cost vector, and $\eps \in (0,0.5)$ be an error parameter. 
    There is an online algorithm  \fullAlg~(\Cref{alg:full-alg}) such that:
    \begin{itemize}
        \item \fullAlg{} computes $f$ to error $O(\eps)$, and 
        \item Assuming $\optworst_\eps = \Omega(1)$, the expected cost is at most
        \[
            \beta n + \optworst_\eps \cdot O\left(\frac{1}{\eps^2}\cdot\log\pbra{\frac{\log \optworst_\eps}{\eps}} \right)\cdot\sum_{i=1}^{n}\Inf_i[f]\pbra{1 + \ln\frac{1}{\Inf_i[f]}}. 
        \]
    \end{itemize}
\end{proposition}

\Cref{thm:iprr-main} follows via an application of Markov's inequality; this observation was also made by Blanc, Lange, and Tan~\cite{blanc2021query}.  

\begin{proof}[Proof of~\Cref{thm:iprr-main} from \Cref{prop:markov}]
    It suffices to show that 
    \begin{equation} \label{eq:almost-free}
        \optworst_{2\eps}(f,c) \leq \frac{\optavg_\eps(f,c)}{\eps}. 
    \end{equation}
    Indeed, \Cref{thm:iprr-main} follows immediately from~\Cref{prop:markov} and~\Cref{eq:almost-free}. 
    
    Let $\calA$ be the $\eps$-error strategy with expected cost $\optavg_\eps$. 
    Consider the strategy $\calA'$ which executes $\calA$, except that it outputs $1$ if it makes a query leading to cost greater than \smash{$\frac{\optavg_\eps}{\eps}$}. 
    By Markov's inequality, note that $\calA$ differs from $\calA'$ with probability at most $\eps$ and consequently is a $2\eps$-error strategy for $f$.   
    Furthermore, note that it has worst-case cost at most $\optavg_\eps/\eps$, and so  \Cref{eq:almost-free} holds.
\end{proof}

The remainder of the section will establish~\Cref{prop:markov}. Our modification of \WIPRR{} is given by the algorithm \IPRR{} (\Cref{alg:full-alg}). 
The algorithm \IPRR{} differs from \WIPRR{} in two ways:
\begin{itemize}
    \item First, the algorithm takes in an additional input in the form of the threshold $B > 0$. 
    We can view $B$ as a proxy for the ``total budget'' of the algorithm. 
    (Looking ahead, our full algorithm \fullAlg{} (\Cref{alg:full-alg}) will use a doubling strategy to guess good choices of $B$ until $B = \optworst_\eps$.) 
    \item Second, the algorithm uses a different exit condition (Step~2(b) above).  
    Recall that the algorithm \WIPRR{} terminated when the bias of the restricted function was sufficiently small. 
\end{itemize}

\begin{algorithm}[t]
    \caption{The $\IPRR$ Algorithm}
    \label{alg:iprr}

    \vspace{1em}
    \textbf{Input:} Succinct representation of $f:\zo^n\to\zo$, threshold $B > 0$, $\eps \in (0, 0.5]$ \\[0.25em]
    \textbf{Output:} A bit $b \in \zo$

    \ 

    $\IPRR(f, \eps, B)$:
    \begin{enumerate}
        \item Initialize $\theta \leftarrow 0^n$ and $\pi \leftarrow \emptyset$. 
        \item Repeat:
        \begin{enumerate}
            \item Let $i^\ast \in [n]$ be the index such that 
            \[
                i^\ast = \arg\max_i \frac{\Inf_i[f_{\pi}]}{\theta_i}\,.
            \]
            \item If $\frac{\Inf_{i^\ast}[f_\pi]}{\theta_{i^\ast}} < \frac{\eps}{B}$, then halt and output $\1{\ex{\bx\sim\zo^n}{f_\pi(\bx)} \ge 1/2}$. \item Else:
            \begin{itemize}
                \item Spend cost $\beta$ towards $x_{i^\ast}$ and update $\theta \leftarrow \theta + \beta e_{i^\ast}$.
                \item If $x_{i^\ast}$ is revealed to be $b_{i^\ast} \in \zo$, then update $\pi \leftarrow \pi \cup \{x_{i^\ast} \mapsto b_{i^\ast}\}$.
            \end{itemize}
        \end{enumerate}
    \end{enumerate}   
\end{algorithm} 

The following proposition guarantees that $\IPRR{}$ has small error for the right choice of threshold $B$:

\begin{proposition}
    \label{prop:thresholded-IPRR-on-correct-cost}
    Let $f\isazofunc$ and $c$ be an unknown cost vector.  
    When $B \geq \optworst_\eps$, the algorithm $\IPRR = \IPRR(f, \eps, B)$ is a $2\eps$-error algorithm for $f$. 
\end{proposition}

\begin{proof}
    Let $T$ be the $\eps$-error algorithm witnessing $\optworst_\eps$, i.e. 
    \[
        \error_f(T) \le \eps
    \quad\text{and}\quad
        \costworst_c(T) = \optworst_\eps(f, c).
    \]
    As in the proof of \Cref{thm:warmup-iprr}, for any restriction $\pi$, we let $T_{\pi}$ denote the algorithm obtained from $T$ by enforcing the restriction $\pi$: $T_{\pi}$ follows $T$ and, whenever $T$ queries a variable $x_i$ such that ``$x_i \mapsto b_i$'' is in the restriction $\pi$, $T_{\pi}$ simply forwards $b_i$ to $T$ as the value of $x_i$ and proceeds. By construction, we have
    \[
        \costavg^{f_\pi}_c(T_\pi)
    \le \costworst_c^{f_\pi}(T_\pi)
    \le \costworst_c^f(T).
    \]
    
    Applying \Cref{lemma:OSSS-consequence} to $f_\pi$ and $T_{\pi}$, we have the invariant that for any restriction $\pi$ during the execution of~\Cref{alg:iprr}, 
    \begin{equation} 
    \label{eq:life-alive}
        \max_{i} \frac{\Inf_i[f_\pi]}{c_i} \geq \frac{\bias(f_\pi) - \dist(T_\pi, f_\pi)}{\costavg^{f_\pi}_c(T_\pi)}.  
    \end{equation}
    Note that~\Cref{alg:iprr} only halts in Step~2(b) when 
    \[
        \frac{\eps}{B} \geq \max_{i} \frac{\Inf_i[f_\pi]}{c_i}\,.
    \]
    Since $B \geq \optworst_\eps(f, c) = \costworst^f_c(T)$, it follows that when the algorithm halts, the restriction $\pi$ satisfies 
    \[
        \frac{\eps}{\costworst^f_c(T)}
    \ge \frac{\eps}{B}
    \geq \frac{\bias(f_\pi) - \dist(T_\pi, f_\pi)}{\costavg^{f_\pi}_c(T_\pi)} \geq \frac{\bias(f_\pi) - \dist(T_\pi, f_\pi)}{\costworst^f_c(T)}\,,
    \]
    which can be rearranged to $\bias(f_\pi) \leq \eps + \dist(T_\pi, f_\pi)$. 
    As in the proof of~\Cref{thm:warmup-iprr}, we let $\bpi$ be the restriction induced by the algorithm~\IPRR{} on a uniformly random input $\bx\sim\zo^n$.
    Taking expectations, we get that 
    \[
        \error_f(\IPRR) = \Ex_{\bx\sim\zo^n}\sbra{\bias(f_{\bpi})} \leq \eps + \Ex_{\bx\sim\zo^n}\sbra{\dist(T_{\bpi}, f_{\bpi})} \leq 2\eps\,,
    \]
    where the final inequality relies on the fact that $T$ is an $\eps$-error algorithm for $f$.
\end{proof}

\Cref{prop:thresholded-IPRR-on-correct-cost} suggests a natural strategy: double our guess for $B$ until we hit $\optworst_\eps$. 
Note, however, that we do not know if our current guess is correct. 
We remedy this by \emph{testing} if our budget $B$ gives an empirical error of $O(\eps)$; this in done in Step~2(b) of \Cref{alg:full-alg}.

\begin{algorithm}[t]
    \caption{The \fullAlg~Algorithm}
    \label{alg:full-alg}

    \vspace{1em}
    \textbf{Input:} Succinct representation of $f:\zo^n\to\zo$, $\eps \in (0, 0.5]$ \\[0.25em]
    \textbf{Output:} A bit $b \in \zo$

    \ 

    $\fullAlg(f,\eps)$:
    \begin{enumerate}
        \item Initialize $i \leftarrow 0$. 
        \item Repeat:
        \begin{enumerate}
            \item Increment $i \gets i + 1$. Set $B_i := 2^i$ and $m_i := \Theta\pbra{\frac{1}{\eps}\log\frac{i}{\eps}}$. 
            \item Draw $\bx^{(1)}, \dots, \bx^{(m_i)} \sim \zo^n$ and let $\bb^{(\ell)} \leftarrow \IPRR(f, \eps, B_i)$ by simulating $\IPRR$ \newline on  $\bx^{(\ell)}$ as the unknown input.
            \item Compute 
            \[
                \bT_i := \frac{1}{m_i}\sum_{\ell=1}^{m_i} \mathbf{1}\cbra{\bb^{(\ell)} \neq f(\bx^{(\ell)})}\,,
            \]
            and if $\bT_i \leq 3\eps$, run $\IPRR(f,\eps,B_i)$ on the unknown input $x$. 
            
        \end{enumerate}
    \end{enumerate}   
\end{algorithm} 

\newcommand{\thetaout}{\theta^{\textrm{outer}}}
\newcommand{\thetain}{\theta^{\textrm{inner}}}
\begin{remark}\label{remark:simulation}
    We provide more details on how \fullAlg{} (\Cref{alg:full-alg}) simulates the \IPRR{} algorithm in Step~2(b). For clarity, let $\thetaout$ (resp., $\thetain$) denote the cumulative investment vector maintained by \fullAlg{} (resp., a simulation of \IPRR{}). Also, let $\bx^{(\ell)}$ denote the input for the simulation, and let $\bx$ denote the actual unknown input on which \fullAlg{} evaluates $f$. Recall that, in the simulations in Step~2(b), the input $\bx^{(\ell)}$ is \emph{known} to \fullAlg{} (and unknown to \IPRR{}).
    
    Whenever the simulated call to \IPRR{} increments $\thetain_i$ for some $i \in [n]$, \fullAlg{} handles it differently in the following three cases:
    \begin{itemize}
        \item \textbf{Case 1:} \fullAlg{} has already observed the value of $x_i$ (by reaching $\thetaout_i \ge c_i$). In this case, we know the value of $c_i$. Then, if the incremented value of $\thetain_i$ also reaches $c_i$, we reveal the value of $x^{(\ell)}_i$ to \IPRR{}; otherwise, we do nothing.
        \item \textbf{Case 2:} \fullAlg{} has not observed $x_i$ and $\thetain_i \le \thetaout_i$. We would know that $\thetain_i \le \thetaout_i < c_i$, so $x^{(\ell)}_i$ should not be revealed to \IPRR{} at this moment. We do nothing.
        \item \textbf{Case 3:} \fullAlg{} has not observed $x_i$ and $\thetain_i > \thetaout_i$ after the increment. In this case, we actually increment $\thetaout_i$ to match $\thetain_i$. If $x_i$ gets revealed to us as a result, we would know that $\thetain_i = \thetaout_i \ge c_i$, so we reveal the value of $x^{(\ell)}_i$ to \IPRR{}.
    \end{itemize}

    The simulation above ensures that, from the perspective of \IPRR{}, it indeed runs on an instance with unknown cost vector $c$ and input $\bx^{(\ell)}$. Note that this simulation does not require the full knowledge of $c$. Also, we always match the investment made by \IPRR{} whenever $\thetain_i$ exceeds $\thetaout_i$. As a result, the cumulative investment of \fullAlg{} in each variable $x_i$ is given by the \emph{maximum} (rather than the sum) over all cumulative investments into $x^{(\ell)}_i$ made by different calls to \IPRR{} throughout the execution of \fullAlg{}.
\end{remark}

A standard Chernoff bound gives us the following guarantee on the testing in Step~2(b):
\begin{lemma} 
\label{lemma:testing}
    Let $\IPRR = \IPRR(f,\eps,B_i)$. We have the following:
    \begin{itemize}
        \item If $\dist(\IPRR, f) \geq 4\eps$, then $\Pr[\bT_i > 3\eps]$ with probability at least $1 - \Theta(\eps i^{-2})$. 
        \item If $\dist(\IPRR, f) \leq 2\eps$, then $\Pr[\bT_i \leq 3\eps]$ with probability at least $1 - \Theta(\eps i^{-2})$. 
    \end{itemize}     
\end{lemma}

\begin{proof}
    This is an easy consequence of a standard (multiplicative) Chernoff bound; see, for e.g., Exercise~2.3.5 of~\cite{vershynin2018high}. 
    In particular, for i.i.d.~Bernoulli$(p)$ random variables $\bX_1, \dots, \bX_m$ for $t\in(0,1)$, we have 
    \[
         \Prx\sbra{\abs{\sum_{\ell=1}^m \bX_\ell - pm} \geq tpm} \leq 2e^{-\Omega(t^2pm)}\,.
    \]


    To see Item~1, let $p := \dist(\IPRR, f) \geq 4\eps$ by assumption. 
    Let $\bX_\ell := \mathbf{1}\{\bb^{(\ell)} \neq f(\bx^{(\ell)})\}$ and note that $\bX_\ell\sim$ Bernoulli$(p)$. 
    We then have 
    \begin{align*}
        \Prx\sbra{\bT_i \leq 3\eps} &= \Prx\sbra{\sum_{\ell=1}^{m_i}  \bX_\ell \leq 3\eps m_i} \\
        &\leq  \Prx\sbra{\abs{ \sum_{\ell=1}^{m_i} \bX_\ell - p m_i} \geq 0.01 pm_i} \\ 
        &\leq 2e^{-\Omega(\eps m_i)} \leq \Theta\pbra{\frac{\eps}{i^2}}\,,\tag{Since $m_i = \Theta\pbra{\eps^{-1}\log(\eps^{-1}i)}$}
    \end{align*}
    where we relied on the fact that $p \geq 4\eps$ to choose $t = 0.01$ while applying the Chernoff bound. 
    The proof of Item~2 is identical. 
\end{proof}

Finally, we turn to the proof of~\Cref{prop:markov}, which in turn completes the proof of~\Cref{thm:iprr-main}:

\begin{proof}[Proof of~\Cref{prop:markov}]
    Let random variable $\bi^\ast$ denote the value of the counter $i$ when the algorithm calls \IPRR{} in Step~2(c). Let $i_0 \ge 1$ be the smallest positive integer such that $2^{i_0} \ge \optworst_\eps$. Note that, assuming $\optworst_\eps = \Omega(1)$, we have $2^{i_0} = O(\optworst_\eps)$.

    We first show that \fullAlg{} computes $f$ to error $O(\eps)$. Let $E$ be the event that the two conditions in \Cref{lemma:testing} simultaneously hold for all $i \in [i_0]$. By \Cref{lemma:testing} and the union bound, we have
    \[
        \Prx[\neg E] \le \sum_{i=1}^{i_0}\Theta(\eps i^{-2}) \le O(\eps)\,.
    \]
    Moreover, we note that event $E$ implies the following two conditions:
    \begin{itemize}
        \item $\bi^{\ast} \le i_0$, i.e., \fullAlg{} must run \IPRR{} in Step~2(c) using one of the first $i_0$ guesses $B_1, B_2, \ldots, B_{i_0}$. This is because, by definition of $i_0$, we have $B_{i_0} = 2^{i_0} \ge \optworst_\eps$. By \Cref{prop:thresholded-IPRR-on-correct-cost}, $\IPRR(f, \eps, B_{i_0})$ computes $f$ to error $2\eps$. Event $E$ then ensures that, if \fullAlg{} does not halt in the first $i_0 - 1$ iterations, the condition $\bT_{i_0} \le 3\eps$ would be true, in which case \fullAlg{} runs \IPRR{} in Step~2(c) with parameter $B_{i_0}$.
        \item $\error_f(\IPRR(f, \eps, B_{\bi^\ast})) \le 4\eps$, i.e., when \fullAlg{} runs \IPRR{} in Step~2(c), the parameter $B_{\bi^\ast}$ ensures that \IPRR{} is a $4\eps$-error algorithm. This is because, for \fullAlg{} to run $\IPRR$ in Step~2(c), the condition $\bT_{\bi^\ast} \le 3\eps$ must hold. Event $E$ then ensures that $\IPRR(f, \eps, B_{\bi^\ast})$ computes $f$ to error $4\eps$.
    \end{itemize}
    The two observations above show that, conditioned on the event $E$, the error probability is bounded by $4\eps$. Thus, we have that 
    \begin{align*}
        \error_f(\fullAlg) 
        &\leq \Pr[\neg E] + 4\eps\cdot\Pr[E] = O(\eps)\,.
    \end{align*}

    Next, we turn to controlling the expected cost of \fullAlg. It suffices to show that, for every $j \in [n]$, the expected investment of \fullAlg{} in $x_j$ is upper bounded by
    \begin{equation}\label{eq:fullAlg-single-investment}
        \beta + \Inf_j[f] \cdot \pbra{1 + \ln\frac{1}{\Inf_j[f]}}\cdot O\pbra{\frac{\optworst_\eps}{\eps^2}\log\pbra{\frac{\log\optworst_\eps}{\eps}}}.
    \end{equation}
    The proposition would directly follow from summing the above over $j \in [n]$.

    Towards proving the upper bound in \Cref{eq:fullAlg-single-investment}, we fix $j \in [n]$ and examine the expected investment in $x_j$ when $\IPRR(f, \eps, B_i)$ runs on a uniformly random input $\bx \sim \zo^n$. The analysis will be similar to one from the proof of~\Cref{thm:warmup-iprr}, so we will be succinct. 
    On any fixed input $x$, let $\Pi(x)$, $\pi^{(j)}$ and $\theta^\ast_j$ be as in the proof of~\Cref{thm:warmup-iprr}. 
    Given that the algorithm has not yet halted as a result of Step~2(b) before incrementing $\theta_j$ from $\theta^\ast_j - \beta$ to $\theta^\ast_j$, we have 
    \[
        \frac{\Inf_j[f_{\pi^{(j)}}]}{\theta^\ast_j - \beta} \geq \frac{\eps}{B}\,,
        \qquad\text{and so}\qquad 
        \theta^\ast_j \leq \beta + \frac{B}{\eps}\cdot\Inf_j[f_{\pi^{(j)}}] \leq \beta + \frac{B}{\eps}\cdot\max_{\pi \in \Pi(x)}\Inf_j[f_\pi]\,.
    \]

    For each $i \in [\bi^\ast]$, let
    \[
        \bY_i \coloneqq \max_{\pi \in \Pi(\bx^{(1)}) \cup \cdots \cup \Pi(\bx^{(m_i)}) \cup \Pi(\bx)}\Inf_j[f_\pi]
    \]
    denote the maximum influence $\Inf_j[f_\pi]$ when $\pi$ ranges over all restrictions encountered by $\IPRR(f, \eps, B_i)$ when it runs on $\bx^{(1)}, \bx^{(2)}, \ldots, \bx^{(m_i)}$ (and possibly the unknown input $\bx$) using the $i$-th guess $B_i$. Then, the maximum investment in $x_j$ made by any call to $\IPRR(f, \eps, B_i)$ is upper bounded by
    \[
        \beta + \frac{B_i}{\eps}\cdot\bY_i.
    \]

    Recall from \Cref{remark:simulation} that the cumulative investment made by \fullAlg{} in each $x_j$ is exactly the maximum investment in $x_j$ among all calls to \IPRR{} during its execution. Therefore, the expected investment in $x_j$ from \fullAlg{} is upper bounded by
    \begin{equation}\label{eq:fullAlg-single-investment-1}
        \ex{}{\beta + \max_{i \in [\bi^\ast]}\frac{B_i}{\eps}\cdot\bY_i}
    \le \beta + \ex{}{\sum_{i=1}^{\bi^\ast}\frac{B_i}{\eps}\cdot\bY_i}
    =   \beta + \sum_{i=1}^{\infty}\frac{B_i}{\eps} \cdot \ex{}{\bY_i \mid \bi^\ast \ge i} \cdot \pr{}{\bi^\ast \ge i}.
    \end{equation}

    To control the last summation in the above, we will prove the following two claims:
    \begin{itemize}
        \item For every $i \ge 1$, $\ex{}{\bY_i \mid \bi^\ast \ge i} \le (m_i + 1)\cdot\Inf_j[f] \cdot \pbra{1 + \ln\frac{1}{\Inf_j[f]}}$. 
        \item For every $i > i_0$, $\pr{}{\bi^\ast \ge i} \le \frac{1}{4^{i - i_0}}$.
    \end{itemize}
    Assuming the two claims above, the upper bound in \Cref{eq:fullAlg-single-investment} follows from a straightforward calculation: Plugging the upper bound on $\ex{}{\bY_i \mid \bi^\ast \ge i}$ as well as $B_i = 2^i$ and $m_i = O(\eps^{-1}\log(i/\eps))$ into the last summation in \Cref{eq:fullAlg-single-investment-1} gives
    \begin{align*}
        &~\sum_{i=1}^{\infty}\frac{B_i}{\eps}\cdot (m_i + 1)\cdot\Inf_j[f] \cdot \pbra{1 + \ln\frac{1}{\Inf_j[f]}} \cdot \pr{}{\bi^\ast \ge i}\\
    \le &~O\pbra{\frac{1}{\eps^2}}\cdot\Inf_j[f] \cdot \pbra{1 + \ln\frac{1}{\Inf_j[f]}}\cdot\sum_{i=1}^{\infty}2^i\log\frac{i}{\eps}\cdot\pr{}{\bi^\ast \ge i}.
    \end{align*}
    The summation $\sum_{i=1}^{\infty}2^i\log\frac{i}{\eps}\cdot\pr{}{\bi^\ast \ge i}$ can be further rewritten as
    \begin{align*}
        &~\sum_{i=1}^{i_0}2^i\log\frac{i}{\eps}\cdot\pr{}{\bi^\ast \ge i} + \sum_{i=i_0 + 1}^{\infty}2^i\log\frac{i}{\eps}\cdot\pr{}{\bi^\ast \ge i}\\
    \le &~\sum_{i=1}^{i_0}2^i\log\frac{i}{\eps} + \sum_{i=i_0 + 1}^{\infty}2^i\log\frac{i}{\eps}\cdot\frac{1}{4^{i - i_0}}\\
    =   &~O\pbra{2^{i_0}\log\frac{i_0}{\eps}}
    =  O\pbra{\optworst_\eps\log\pbra{\frac{\log\optworst_\eps}{\eps}}}.
    \end{align*}
    The first step above upper bounds $\pr{}{\bi^\ast \ge 1}$ by $1$ when $i \le i_0$, and by $1/4^{i - i_0}$ when $i > i_0$. The second step observes that the two summations are dominated by the terms at $i = i_0$ and $i = i_0 + 1$, respectively. The last step applies $2^{i_0} = O(\optworst_\eps)$, which follows from the definition of $i_0$.

    It remains to verify the two claims regarding $\ex{}{\bY_i \mid \bi^\ast \ge i}$ and $\pr{}{\bi^\ast \ge i}$. For the first claim, note that, conditioning on $\bi^\ast \ge i$, the $m_i$ inputs $\bx^{(1)}, \ldots, \bx^{(m_i)}$ as well as the unknown input $\bx$ are still uniformly distributed over $\zo^n$. Proceeding \emph{mutatis mutandis} as in the martingale argument from the proof of~\Cref{thm:warmup-iprr}, we can show that, when $\IPRR(f, \eps, B_i)$ runs on each of these $m_i + 1$ inputs, the expected maximum value of $\Inf_j[f_\pi]$ is at most $\Inf_j[f] \cdot\pbra{1 + \ln\frac{1}{\Inf_j[f]}}$. Recall that $\bY_i$ is defined as the maximum over these $m_i + 1$ maxima. By relaxing the maximum to a sum, we obtain the claim that $\ex{}{\bY_i \mid \bi^\ast \ge i} \le (m_i + 1)\Inf_j[f] \cdot\pbra{1 + \ln\frac{1}{\Inf_j[f]}}$.

    For the second claim, note that for every $i \ge i_0$, it holds that $B_i \ge 2^{i_0} \ge \optworst_\eps$. Then, \Cref{prop:thresholded-IPRR-on-correct-cost} and \Cref{lemma:testing} together imply that \fullAlg{} would halt in Step~2(c) except with probability $O(\eps i^{-2})$. By carefully choosing the hidden constant factor in $O(\cdot)$, this probability is at most $1/4$. Then, in order for $\bi^\ast \ge i \ge i_0 + 1$ to happen, \fullAlg{} must fail to halt on each of the $i - i_0$ guesses $B_{i_0}, B_{i_0 + 1}, \ldots, B_{i-1}$, which happens with probability at most $1/4^{i - i_0}$.

    This verifies the two claims, thus establishing that \Cref{eq:fullAlg-single-investment} upper bounds the expected investment in $x_j$ and proving the proposition.

\end{proof}

\subsection{A Sharper Analysis for Symmetric Functions}
\label{subsec:symmetric}

Recall that a Boolean function $f\isazofunc$ is \emph{symmetric} if the following holds:
\[
    f(x) = f(y) \qquad\text{whenever}\qquad \sumi x_i = \sumi y_i\,.
\]
Note that many natural functions, including $\AND$, $\OR$, and $\MAJ$ are symmetric. 
Furthermore, prior work on priced query strategies has designed algorithms tailored to symmetric functions~\cite{gkenosis2018stochastic,hellerstein2018stochastic,hellerstein2022adaptivity,hellerstein2024quickly}. 

We show that the algorithm \WIPRR{} achieves the following performance guarantee: 

\begin{theorem}[Formal version of \Cref{thm:symmetric-informal}]\label{thm:symmetric}
    For any symmetric function $f\isazofunc$, the algorithm $\WIPRR(f,\eps)$ with unit investment $\beta$ is an $\eps$-error algorithm with expected cost 
    \[
        \beta n + \optavg_0\cdot O\pbra{\log\frac{1}{\eps}}. 
    \]
\end{theorem}

\paragraph{Necessity of the log factor.} To see why the $\log(1/\eps)$ factor is necessary, consider the case where $f$ is the $\AND$ function over $n = \Theta(\log(1/\eps))$ variables, each with a unit cost. The optimal algorithm simply queries the variables in order, and stops as soon as any of the input bits is revealed to be a zero. The expected cost is then $\optavg_0 = O(1)$. On the other hand, the \WIPRR{} algorithm needs to pay the cost of all the $n$ variables before seeing any input bit, leading to a gap of $\Omega(n) = \Omega(\log(1/\eps))$. Note that we cannot take $n$ to be larger than $\log(1/\eps)$; otherwise $f$ would be $\eps$-close to the constant zero function and \WIPRR{} simply outputs $0$ without incurring any cost. In \Cref{sec:tribes-lb}, we show that \emph{any} online algorithm needs to pay this additional $\log(1/\eps)$ factor.

\newcommand{\btau}{\boldsymbol{\tau}}

\subsubsection{Analysis of the Optimal Algorithm and Round-Robin}
For brevity, we rename the input bits in increasing order of their costs, i.e., $c_1 \le c_2 \le \cdots \le c_n$. Below are the behavior and the associated costs of the optimal (offline) algorithm and the \WIPRR{} algorithm:
\begin{itemize}
	\item The offline algorithm queries $x_1, x_2, \ldots, x_n$ in order, and stops as soon as for some $t$, $f|_{x_1,\ldots,x_t}$ becomes a constant function. Let $\btau^{(0)}$ denote the stopping time of the algorithm, i.e., the smallest such index $t$ on a uniformly random input $\bx\sim\zo^n$. 
    The expected cost is given by
	\begin{equation}\label{eq:opt-symmetric}
		\optavg_0(f, c) = \sum_{i=1}^{n}c_i \cdot \pr{}{\btau^{(0)} \ge i}\,.
	\end{equation}
	\item When running \WIPRR{}, the input bits $x_1, x_2, \ldots, x_n$ are revealed in order, and the algorithm stops as soon as for some $t$, $f|_{x_1,\ldots,x_t}$ becomes $\eps$-close to a constant function. 
    Let $\btau^{(\eps)}$ denote this stopping time for $\bx\sim\zo^n$. 
    The expected cost of \WIPRR{} is then given by
	\begin{equation}\label{eq:sol-symmetric}
		\costavg^f_c(\WIPRR) \le \beta n + \sum_{i=1}^{n}c_i \cdot \left[\pr{}{\btau^{(\eps)} \ge i} + (n - i)\cdot\pr{}{\btau^{(\eps)} = i}\right]\,.
	\end{equation}
	Here, the term $(n - i)\cdot\pr{}{\btau^{(\eps)} = i}$ is due to that if $\bx_1$ through $\bx_{\tau^{(\eps)}}$ get revealed by the end of the algorithm, we also pay a cost of $c_{\tau^{(\eps)}}$ on each of the input bits $\bx_{\tau^{(\eps)} + 1}$ through $x_n$. 
\end{itemize}

See~\Cref{fig:symm} for an illustration of this.

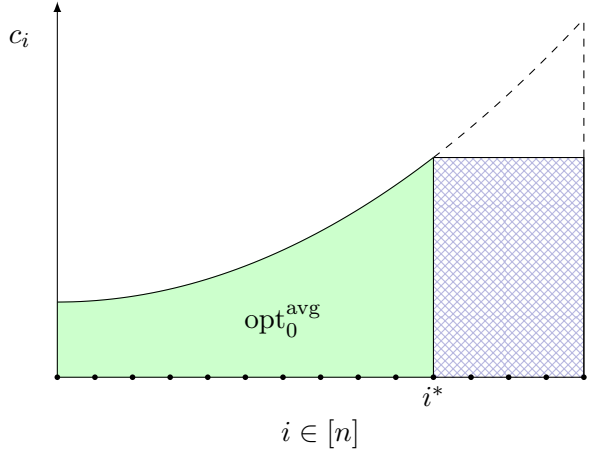
\begin{figure}
    \centering
    \begin{tikzpicture}

        \def\lambdaPlot{\x,{1 + (\x^2/13)}}
        \filldraw[draw,pattern=crosshatch,pattern color=blue!40!gray!40] (5,0) -- (5,2.923) -- (7,2.923) -- (7,0);
        \filldraw[draw,color=black,green!20] (0,0) -- plot[domain=0:5,smooth] (\lambdaPlot) -- (5,2.923) -- (5,0) -- (0,0);
        
        \draw[-latex] (0,0) -- (0,5);
        \draw (0,0) -- (7,0);

        \draw (7,0) -- (7,2.923);
        \draw[dashed] (7,2.923) -- (7,4.769);
        \node (iast) at (5,-0.25) {$i^\ast$};

        \draw (5,0) -- (5,2.923);

        \node (opt) at (3,0.75) {$\optavg_0$};

        \foreach \x in {0,1,...,14} {
            \node at (\x/2,0) [circle,fill,inner sep=.75pt]{};
        }

        \node (ex) at (3.5,-0.75) {$i\in[n]$};
        \node (why) at (-0.5,4.5) {$c_i$};
        \draw[] plot[scale=1,domain=0:5,smooth] (\lambdaPlot);
        \draw[dashed] plot[scale=1,domain=5:7,smooth] (\lambdaPlot);
    \end{tikzpicture}
    \caption{An illustration of the behaviors of the optimal offline algorithm and $\WIPRR$ on a symmetric function. Here, we assume that the costs are ordered as $c_1 \leq c_2 \leq \cdots \leq c_n$, and $i^\ast$ denotes the index s.t.~$\bias(f) \leq\eps$ after fixing $x_1, \dots, x_{i^\ast}$. The cross-hatched region denotes the additional cost incurred by $\WIPRR$ in contrast to the optimal offline algorithm.}
    \label{fig:symm}
\end{figure}

\subsubsection{Proof of \Cref{thm:symmetric}}
In light of \Cref{eq:opt-symmetric,eq:sol-symmetric}, to prove \Cref{thm:symmetric}, it suffices to show that for any symmetric function $f$ and any set of costs $0 \le c_1 \le c_2 \le \cdots \le c_n$, it holds that
\begin{equation}\label{eq:symmetric-general-costs}
	\sum_{i=1}^{n}c_i \cdot \left[\pr{}{\btau^{(\eps)} \ge i} + (n - i)\cdot\pr{}{\btau^{(\eps)} = i}\right]
\le O\pbra{\log\frac{1}{\eps}} \cdot \sum_{i=1}^{n}c_i \cdot \pr{}{\btau^{(0)} \ge i}.
\end{equation}

Note that it is sufficient to prove the above for costs of form $c_i = \mathbf{1}\cbra{i \geq \istar}$ for some $\istar \in [n]$, i.e., there is a unit cost for every index above the threshold $\istar$, and the cost is zero below $\istar$. This is because: (1) every increasing sequence of costs $c_1 \le c_2 \le \cdots \le c_n$ can be written as a conic (non-negative) combination of costs of form $c_i = \mathbf{1}\cbra{i \ge \istar}$; (2) both sides of \Cref{eq:symmetric-general-costs} are linear in $c$.

With the observation above, it remains to show that, for any symmetric $f$ and $\istar \in [n]$, it holds that
\[
	(n - \istar + 1)\cdot\pr{}{\btau^{(\eps)} \ge \istar} \le O(\log(1/\eps)) \cdot \sum_{i = \istar}^{n}\pr{}{\btau^{(0)} \ge i},
\]
which is equivalent to 
\begin{equation}\label{eq:symmetric-step-costs}
	(n - \istar + 1)\cdot\pr{}{\btau^{(\eps)} \ge \istar} \le O(\log(1/\eps)) \cdot \ex{}{\max\{\btau^{(0)} - \istar + 1, 0\}}.
\end{equation}

We will prove an even stronger inequality than the above---after conditioning on the realization of $x_1, \ldots, x_{\istar - 1}$, the inequality above still holds. Note that the input bits $x_1, \ldots, x_{\istar - 1}$ determines whether $\btau^{(\eps)} < \istar$ or $\btau^{(\eps)} \ge \istar$. In the former case, the contribution to the left-hand side of \Cref{eq:symmetric-step-costs} is zero, so the inequality trivially holds. In the remaining case where $\tau^{(\eps)} \ge \istar$ given $x_1$ through $x_{\istar - 1}$, the left-hand side of \Cref{eq:symmetric-step-costs} reduces to $n - \istar + 1$. The right-hand side, on the other hand, is simply $O(\log(1/\eps))$ times the expectation of the stopping time ``$\btau^{(0)}$'' defined with respect to the symmetric function $f|_{x_1, \ldots, x_{\istar - 1}}$ over $n - \istar + 1$ variables.

Therefore, it remains to prove the following lemma:
\begin{lemma}\label{lemma:tau-0-lower-bound}
	For any symmetric function $f$ on $n$ variables with $\bias(f) \ge \eps$, it holds that
	\[
		\ex{}{\btau^{(0)}} \ge \Omega\left(\frac{n}{\log(1/\eps)}\right),
	\]
	where the stopping time $\tau^{(0)}$ (defined over the uniform drawing of $x_1, x_2, \ldots, x_n$) is the smallest value $t$ such that $f|_{x_1, \ldots, x_t}$ reduces to a constant.
\end{lemma}

\begin{proof}
	If $\eps \le e^{-n/100}$, we only need to show that $\ex{}{\btau^{(0)}} \ge \Omega(1)$, which is trivially true since $\btau^{(0)}$ is at least $1$. Otherwise, we claim that $\btau^{(0)}$ must be at least $n/2$ with probability $\Omega(1)$, which gives the stronger bound of $\ex{}{\btau^{(0)}} \ge \Omega(n)$.

	To see this, note that after $n/2$ input bits are revealed, with probability $\Omega(1)$, both $0$ and $1$ appear $\le n/3$ times. If $f$ reduces to a constant function, $f$ must be constant on every input $x$ with Hamming weight in $[n/3, 2n/3]$, which, by a standard Chernoff bound, contradicts the assumption that $f$ is at distance $\eps \ge e^{-n/100}$ from constant functions. 
\end{proof}

\section{Query Strategies for Functions with Shallow Decision Trees}
\label{sec:follow-the-tree}

In this section, we turn to the setting of \Cref{thm:follow-the-tree-general-informal}, where the function can be represented as a shallow decision tree, either given to the algorithm or unknown. We begin by considering the special case where the decision tree is given and \emph{everywhere-influential}, meaning that each internal node queries a variable with sufficiently high influence. Then, we deal with the more general case, where we first apply the pruning lemma of~\cite{BLQT22} to transform the tree into an everywhere-influential one, and then follow the pruned tree. When the decision tree representation is not given, we apply the learning algorithm of~\cite{BLQT22} to learn an everywhere-influential decision tree that approximates the function, and then query the variables according to the learned tree.

\subsection{Follow the Everywhere-Influential Tree}\label{sec:follow-the-tree-special}
We adopt the definition of everywhere $\tau$-influential decision trees in~\cite{BLQT22}. In the rest of this section, for each internal node $v$ in a decision tree, we let $\ind(v) \in [n]$ denote the index of the variable queried by $v$.

\begin{definition}[Everywhere $\tau$-influential]\label{def:everywhere-influential}
    For $\tau \in [0, 1]$, a decision tree $T$ is everywhere $\tau$-influential with respect to a function $f: \zo^n \to \zo$ if, for every internal node $v$ in $T$,
    \[
        \Inf_{\ind(v)}[f_v] \ge \tau,
    \]
    where $f_v$ is the restriction of $f$ induced by the root-to-$v$ path in $T$.
\end{definition}

We start by analyzing the most straightforward algorithm when a decision tree representation is given---the algorithm simply follows the tree.

\begin{definition}[Follow the tree]\label{def:follow-the-tree}
    Given the decision tree $T$, the algorithm starts at the root of $T$. At each internal node $v$, the algorithm increments $\theta_{\ind(v)}$ by $\beta$ until it reaches the cost $c_{\ind(v)}$, at which point $x_{\ind(v)}$ is revealed. The algorithm follows the computation path of $T$ until a leaf node is reached, at which point the algorithm outputs the label of that leaf.
\end{definition}

Compared to the \IPRR{} algorithms, ``follow the tree'' blindly trusts the given decision tree---at any point, the algorithm concentrates its investment on the variable specified by the tree, rather than hedging among multiple variables. Note that, if the computation path of $T$ queries a variable $x_i$, the investment $\theta_i$ made by ``follow the tree'' is at most $c_i + \beta$, where $\beta$ is the unit investment of the algorithm (cf.~\Cref{def:model}). This leads to the additional $\beta n$ term in all of our upper bounds. Recall that the algorithm may pick a sufficiently small $\beta = 1/\poly(n)$ to make the $\beta n$ term negligible.

When the given decision tree representation of $f$ is everywhere $\tau$-influential (with respect to $f$), ``follow the tree'' is $(1/\tau)$-competitive against the zero-error average-case benchmark.

\begin{proposition}\label{prop:follow-the-tree-special}
    For any $\tau \in (0, 1]$, when ``follow the tree'' is given an everywhere $\tau$-influential decision tree representation of the function, the expected cost of the algorithm is at most
    \[
        \beta n + \optavg_0 \cdot \frac{1}{\tau}.
    \]
\end{proposition}

\Cref{prop:follow-the-tree-special} follows from the two simple lemmas below. The first lemma lower bounds the offline benchmark $\optavg_0$ by a cost-weighted sum of the influences.

\begin{lemma}\label{lemma:optavg-lower-bound}
    For any function $f: \zo^n \to \zo$, the zero-error average-case benchmark satisfies
    \[
        \optavg_0 \ge \sum_{i=1}^{n}c_i \cdot \Inf_i[f].
    \]
\end{lemma}

\begin{proof}
    Let $\calA$ be a zero-error offline algorithm computing function $f: \zo^n \to \zo$. We claim that, when $\calA$ runs on an input $\bx$ that satisfies $f(\bx) \ne f(\bx^{\oplus i})$ (i.e., variable $x_i$ is pivotal on $\bx$), $x_i$ must be revealed to $\calA$. Otherwise, $\calA$ would have the same behavior on the alternative input $\bx^{\oplus i}$ with a different function value, and thus cannot have a zero error. Therefore, we have
    \[
        \pr{\bx\sim\zo^n}{\calA\text{ reveals }x_i}
    \ge \pr{\bx\sim\zo^n}{f(\bx) \ne f(\bx^{\oplus i})}
    =   \Inf_i[f].
    \]
    It follows that, under every cost vector $c \in \R_{\geq 0}^n$, the expected cost of $\calA$ is at least
    \[
        \sum_{i=1}^{n}c_i \cdot \pr{\bx\sim\zo^n}{\calA\text{ reveals }x_i}
    \ge \sum_{i=1}^{n}c_i \cdot \Inf_i[f].
    \]
    This gives the desired lower bound $\optavg_0 \ge \sum_{i=1}^{n}c_i \cdot \Inf_i[f]$.
\end{proof}

The second lemma shows that the influence of each variable $x_i$ is lower bounded by $\tau$ times the probability that it is queried in an everywhere $\tau$-influential decision tree.

\begin{lemma}\label{lemma:influence-lower-bound}
    Suppose that decision tree $T$ is everywhere $\tau$-influential with respect to function $f: \zo^n \to \zo$. Then, for every $i \in [n]$,
    \[
        \Inf_i[f] \ge \tau \cdot \delta_i(T).
    \]
\end{lemma}

\begin{proof}
    Fix an index $i \in [n]$ and consider a uniformly random input $\bx \sim \zo^n$. The goal is to lower bound $\Inf_i[f]$, namely, the probability of $f(\bx) \ne f(\bx^{\oplus i})$. Consider the computation path when $T(\bx)$ is evaluated. Let
    \[
        V \coloneqq \{v: v\text{ is an internal node of }T, \ind(v) = i\}
    \]
    be the set of internal nodes that query $x_i$. For each $v \in V$, let $E_v$ denote the event that node $v$ is reached when computing $T(\bx)$. Note that the events $\{E_v: v \in V\}$ are disjoint, and
    \[
        \sum_{v \in V}\pr{}{E_v} = \delta_i(T).
    \]
    Let $\pi_v$ denote the restriction induced by the root-to-$v$ path. We note that event $E_v$ is equivalent to that the input $\bx$ agrees with restriction $\pi_v$. Therefore, we have
    \begin{align*}
        \pr{\bx\sim\zo^n}{f(\bx) \ne f(\bx^{\oplus i}) \mid E_v}
    &=  \pr{\bx\sim\zo^n}{f(\bx) \ne f(\bx^{\oplus i}) \mid \bx \agrees \pi_v}\\
    &=  \pr{\bx\sim\zo^n}{f_v(\bx) \ne f_v(\bx^{\oplus i})}\\
    &=  \Inf_i[f_v]
    \ge \tau,
    \end{align*}
    where the last step applies the fact that $\ind(v) = i$ and $T$ is everywhere $\tau$-influential with respect to $f$.
    
    Therefore, we conclude that
    \begin{align*}
        \tau \cdot \delta_i(T)
    &=  \sum_{v \in V}\tau \cdot \pr{}{E_v}\\
    &\le\sum_{v \in V}\pr{\bx \sim \zo^n}{f(\bx) \ne f(\bx^{\oplus i}) \mid E_v} \cdot \pr{}{E_v}\\
    &\le\pr{\bx \sim \zo^n}{f(\bx) \ne f(\bx^{\oplus i})} \tag{law of total probability}\\
    &=  \Inf_i[f]\,, \tag{definition of influence}
    \end{align*}
    which completes the proof. 
\end{proof}

\Cref{prop:follow-the-tree-special} is then an immediate consequence of the two lemmas above.

\begin{proofof}{\Cref{prop:follow-the-tree-special}}
    The cost of ``follow the tree'' is given by
    \[
        \sum_{i=1}^{n}(c_i + \beta)\cdot \delta_i(T)
    \le \beta n + \frac{1}{\tau}\sum_{i=1}^{n}c_i \cdot \Inf_i[f]
    \le \beta n + \optavg_0 \cdot \frac{1}{\tau},
    \]
    where the first step applies \Cref{lemma:influence-lower-bound} and the second step applies \Cref{lemma:optavg-lower-bound}.
\end{proofof}

\subsection{Follow the Pruned Tree}\label{sec:follow-the-tree-general}
In the more general case that the algorithm is given an arbitrary decision tree representation $T$ of the function, our strategy is a natural one: First, we ``prune'' the decision tree $T$ into an everywhere $\tau$-influential tree $T'$, so that $T'$ agrees with $T$ on most of the inputs. Such a pruning result was given by~\cite{BLQT22}. Then, we run the ``follow the tree'' algorithm on $T'$.

\paragraph{The pruning lemma of~\cite{BLQT22}.} Let $\Delta(T)$ denote the average depth of decision tree $T$, which is formally defined as:
\[
    \Delta(T) \coloneqq \begin{cases}
        0, & \text{if }T\text{ has only one node},\\
        1 + \frac{1}{2}\left[\Delta(T_0) + \Delta(T_1)\right], & \text{if the root of }T\text{ has subtrees }T_0\text{ and }T_1.
    \end{cases}
\]

We will use the following version of the pruning lemma of~\cite{BLQT22}.

\begin{lemma}[Theorem 4 of \cite{BLQT22}]\label{lemma:pruning-lemma-BLQT}
    For any decision tree $T$ and $\tau \in [0, 1]$, there is an everywhere $\tau$-influential decision tree $T'$ (with respect to the function computed by $T$) such that
    \[
        \dist(T, T') \le \tau \cdot \Delta(T).
    \]
    Furthermore, such a $T'$ can be computed from $T$ in time polynomial in the size of $T$.
\end{lemma}

\begin{definition}[Follow the pruned tree]\label{def:follow-the-pruned-tree}
    Given the decision tree $T$ and accuracy parameter $\eps \in (0, 1]$, the algorithm applies \Cref{lemma:pruning-lemma-BLQT} to $T$ and $\tau \coloneqq \eps / \Delta(T)$ and efficiently samples a pruned tree $T'$. Then, the algorithm runs ``follow the tree'' (\Cref{def:follow-the-tree}) on $T'$.
\end{definition}

\begin{theorem}[Formal version of \Cref{thm:follow-the-tree-general-informal}]\label{thm:follow-the-tree-general}
    Given $\eps \in (0, 1]$ and a decision tree representation $T$ of function $f$, ``follow the pruned tree'' computes $f$ up to an error of $\eps$ and with an expected cost of at most
    \[
        \beta n + \optavg_0 \cdot \frac{\Delta(T)}{\eps}.
    \]
\end{theorem}

\begin{proof}
    Let $\tau \coloneqq \eps / \Delta(T)$ and $T'$ be the pruned tree computed from \Cref{lemma:pruning-lemma-BLQT}. Then, the error of ``follow the pruned tree'' is at most
    \[
        \dist(T, T') \le \tau\cdot\Delta(T) = \eps.
    \]
    Furthermore, the expected cost of the algorithm is given by
    \[
        \sum_{i=1}^{n}(c_i + \beta) \cdot \delta_i(T')
    \le \beta n + \frac{1}{\tau}\sum_{i=1}^{n}c_i \cdot \Inf_i[f]
    \le \beta n + \optavg_0 \cdot \frac{\Delta(T)}{\eps},
    \]
    where the first step applies \Cref{lemma:influence-lower-bound} to $T'$ and $f$, while the second step applies \Cref{lemma:optavg-lower-bound}.
\end{proof}

\begin{remark}
    \label{remark:unknown-tree}
    Suppose that the function $f$ is guaranteed to be computed by some decision tree $T$ of average depth $d = \Delta(T)$, but $T$ is not given. In this case, we can apply the decision tree learning algorithm of~\cite[Theorem 7]{BLQT22} to compute a decision tree $T'$ such that: (1) $\dist(T, T') \le \eps$; (2) $T'$ is everywhere $\tau$-influential with respect to $f$ for $\tau = \Omega(\eps / d)$. The algorithm only requires query access to $f$, and runs in time
    \[
        \poly(n) \cdot \left(\frac{d}{\eps}\right)^{O(d / \eps)}.
    \]
    Then, running ``follow the tree'' on $T'$ gives an error $\le \eps$ and expected cost of at most
    \[
        \beta n + \optavg_0 \cdot O\left(\frac{d}{\eps}\right).
    \]

    Furthermore, when the value of $d = \Delta(T)$ is unknown, we can obtain the same result by guessing the value of $d = 1, 2, 3, \ldots$. For each guess of $d$, we can estimate the error of the tree learned by the algorithm of~\cite{BLQT22} by making membership queries to $f$ on uniformly random inputs. We increment the guess on $d$ until the resulting error is below $\eps$. This gives an algorithm with expected cost $\optavg_0 \cdot O(d/\eps)$ without knowing the value of $\Delta(T)$ in advance.
\end{remark}
\section{Separating the Offline and Online Settings: Proof of \Cref{thm:tribes-lower-bound-informal}}
\label{sec:tribes-lb}

We turn to the proof of~\Cref{thm:tribes-lower-bound}. 

\subsection{Warmup: A Linear Lower Bound against Zero-Error Algorithms}
We start with an $\Omega(n)$ lower bound on the competition ratio, which applies to all \emph{zero-error} algorithms for the $\AND$ function. The costs of the $n$ variables are known to be a permutation of $[n] = \{1, 2, \ldots, n\}$, but the exact ordering of the costs is unknown. Formally, we consider a hard distribution over instances defined as follows.

\begin{definition}[$\AND$ instance]\label{def:and-instance}
    For integer $n \ge 1$, $f(x) = \bigwedge_{i=1}^{n}x_i$ is the $\AND$ function on $n$ variables. The costs $(c_1, c_2, \ldots, c_n)$ are set to a permutation of $[n]$ chosen uniformly at random.
\end{definition}

\begin{proposition}[Formal version of \Cref{prop:and-lower-bound-informal}]
    \label{prop:and-lower-bound}
    On the $\AND$ instance with $n$ variables, we have the offline benchmark $\optavg_0 = O(1)$, while every zero-error online algorithm has an expected cost of $\Omega(n)$. 
\end{proposition}

\begin{proof}
    We first note that $\optavg_0 = O(1)$: Since the offline algorithm knows the costs, it may query the variables in increasing order of costs, and stop whenever a zero is encountered. The algorithm is always correct, and the variable with cost $i$ is queried with probability $2^{-(i-1)}$ (namely, when all the $i-1$ variables with lower costs take value $1$). Thus, the algorithm has an expected cost of
    \[
        \sum_{i=1}^{n}2^{-(i-1)} \cdot i
    \le \sum_{i=1}^{+\infty}2^{-(i-1)} \cdot i
    =	4.
    \]

    It remains to show that every zero-error algorithm must incur an expected cost of $\Omega(n)$. By an averaging argument, it suffices to prove this for deterministic algorithms. This is because any zero-error randomized algorithm can be viewed as a mixture of zero-error deterministic algorithms. If we could lower bound the cost of every zero-error deterministic algorithm by $\Omega(n)$, the same lower bound holds for the mixture as well.

    Fix a deterministic algorithm $\calA$. Let $\calA^{(1)}$ denote the simulation of $\calA$ on an $\AND$ instance with $c_i = n$ for every $i \in [n]$. Note that this is \emph{not} an instance from \Cref{def:and-instance}; later, we will couple this simulation with the execution of $\calA$ on an actual instance. We stop the simulation as soon as either one of the following happens:
    \begin{itemize}
        \item \textbf{Case 1.} $\calA$ ``terminates voluntarily'' by outputting an answer.
        \item \textbf{Case 2.} The total cost, $\sum_{i=1}^{n}\theta_i$, reaches $n/2$. Formally, when $\calA$ attempts to increase some $\theta_i$, we check whether $\sum_{i=1}^{n}\theta_i \ge n/2$ would hold after the increase. If so, we terminate the algorithm before $\theta_i$ gets increased.
    \end{itemize}
    Since $\calA$ is deterministic, when $\calA^{(1)}$ terminates, the investment in each variable $x_i$ is a deterministic value (denoted by $a_i$). We note that $\sum_{i=1}^{n}a_i < n/2$; otherwise, the simulation should have stopped earlier (by Case~2). Since $c_i = n$ for every $i \in [n]$, none of the $n$ variables is revealed before $\calA^{(1)}$ terminates.
    
    Now, suppose that we run $\calA$ on an actual instance in \Cref{def:and-instance}, i.e., the cost vector $c$ is a random permutation of $\{1, 2, \ldots, n\}$. Call this simulation $\calA^{(2)}$. We will show that $\calA^{(2)}$ agrees with $\calA^{(1)}$ with probability $\Omega(1)$. Formally, we have
    \begin{equation}\label{eq:agreement-prob}
        \pr{c}{c_i > a_i,~\forall i \in [n]} \ge \Omega(1).
    \end{equation}
    Assuming \Cref{eq:agreement-prob}, with probability $\Omega(1)$ over the randomness in $c$, $\calA^{(2)}$ agrees with $\calA^{(1)}$ until $\calA^{(1)}$ terminates. Conditioning on this agreement, Case~1---that $\calA^{(1)}$ outputs an answer before the total investment reaches $n/2$---cannot be true. This is because none of the variables has been revealed so far, so $\calA^{(2)}$ would incur a non-zero error by outputting an answer at that time. Thus, we must be in Case~2, i.e., $\calA^{(1)}$ attempts to reach $\sum_{i=1}^{n}\theta_i \ge n/2$ when we stop the simulation. Since $\calA^{(2)}$ agrees with $\calA^{(1)}$, $\calA^{(2)}$ incurs a cost $\ge n/2$ before terminating. This would then imply that $\calA^{(2)}$ has an expected cost of $\Omega(1) \cdot (n/2) = \Omega(n)$.

    To prove \Cref{eq:agreement-prob}, we note that each constraint $c_i > a_i$ gets violated if and only if $c_i$ is in $\{1, 2, \ldots, n\} \cap [0, a_i]$, a set of size at most $a_i$. Over the randomness in the permutation $c$, $c_i$ is uniformly distributed over the size-$n$ set $\{1, 2, \ldots, n\}$. Therefore, we have
	\[
		\pr{c}{c_i \le a_i}
	\le \frac{a_i}{n},
	\]
	and thus, by the union bound,
	\[
		\pr{c}{c_i > a_i,~\forall i \in [n]}
	\ge 1 - \sum_{i=1}^{n}\pr{c}{c_i \le a_i}
	\ge 1 - \sum_{i=1}^{n}\frac{a_i}{n}
	\ge 1 - \frac{n/2}{n}
	=	\frac{1}{2}.
	\]
\end{proof}

\subsection{A Logarithmic Lower Bound against Constant-Error Algorithms}
Even if we allow the algorithm to have a constant error probability, we still have an $\Omega(\log n)$ lower bound on the competitive ratio. We prove this lower bound using the $\Tribes$ function, which we formally define as follows.

\begin{definition}[$\Tribes$ instance]\label{def:tribes-instance}
    For integer $w \ge 1$, $f(x) \coloneqq \bigvee_{i=1}^{2^w}\bigwedge_{j=1}^{w}x_{i,j}$ is the $\Tribes$ function with $2^w$ tribes of width $w$, where we rename the $n = 2^w\cdot w$ variables as $(x_{i,j})_{i \in [2^w], j \in [w]}$ for clarity. For each $i \in [2^w]$, the costs $(c_{i,1}, c_{i,2}, \ldots, c_{i,w})$ are set to a permutation of $[w]$ chosen independently and uniformly at random.
\end{definition}

\begin{theorem}[Formal version of \Cref{thm:tribes-lower-bound-informal}]\label{thm:tribes-lower-bound}
	For every $\eps_0 \in [0, 1/4)$ and integer $w \ge 1$, on the $\Tribes$ instance with width $w$, we have the offline benchmark $\optavg_0 = O(2^w)$, while every $\eps_0$-error online algorithm has an expected cost of $\Omega(2^w \cdot w)$.
\end{theorem}

Since $w = \Theta(\log n)$ in \Cref{def:tribes-instance}, the theorem shows that an online algorithm can be at best $\Omega(\log n)$-competitive compared to the benchmark $\optavg_0$.

We start with an overview of the proof. Let $\calA$ denote an algorithm that computes the width-$w$ $\Tribes$ function accurately. Intuitively, algorithm $\calA$ needs to look at at least one variable in a constant fraction of the $2^w$ tribes. Then, we can transform algorithm $\calA$ into a new one (denoted by $\calA'$) for the $\AND$ instance (\Cref{def:and-instance}) over $w$ variables, so that $\calA'$ reveals one of the $w$ variables with probability $\Omega(1)$. Furthermore, the expected cost of $\calA'$ is only a $2^{-w}$-fraction of that of $\calA$. Finally, we extend our proof of \Cref{prop:and-lower-bound} to show that $\calA'$ must have an expected cost of $\Omega(w)$. This would then lower bound the expected cost of $\calA$ by $2^w \cdot \Omega(w) = \Omega(2^w \cdot w)$.

We start by formally defining an algorithm \emph{``revealing a tribe''} in the context of computing the $\Tribes$ function.

\begin{definition}[Revealing a tribe]\label{def:reveal-tribe}
    When an algorithm runs on a $\Tribes$ instance (\Cref{def:tribes-instance}), we say that the algorithm reveals the $i$-th tribe ($i \in [2^w]$) if at least one of the variables $x_{i,1}, x_{i,2}, \ldots, x_{i,w}$ is revealed to the algorithm.
\end{definition}

\Cref{thm:tribes-lower-bound} follows from the three lemmas below, which we prove in the remainder of this section.

\begin{restatable}{lemma}{numberofrevealedtribes}\label{lemma:number-of-revealed-tribes}
	For every $\eps_0 \in [0, 1/4)$ and integer $w \ge 1$, any $\eps_0$-error algorithm computing the width-$w$ $\Tribes$ function must reveal $\Omega(2^w)$ tribes in expectation.
\end{restatable}

\begin{restatable}{lemma}{fromtribestoand}\label{lemma:from-tribes-to-and}
    Suppose that for some $\alpha, C > 0$, an algorithm for the $\Tribes$ instance with width $w$: (1) reveals at least $\alpha \cdot 2^w$ tribes in expectation; and: (2) has an expected cost of $C$. Then, there is an algorithm for the $\AND$ instance with $w$ variables that: (1) reveals at least one of the variables with probability $\alpha$; and: (2) has an expected cost of $C/2^w$.
\end{restatable}

\begin{restatable}{lemma}{costofseeingavariable}\label{lemma:cost-of-seeing-a-variable}
    For any $p \in (0, 1]$, there exists $\alpha > 0$ such that the following is true for every integer $n \ge 1$: If an algorithm for the $\AND$ instance with $n$ variables reveals at least one of the variables with probability at least $p$, the expected cost of the algorithm must be at least $\alpha \cdot n$.
\end{restatable}

\begin{proof}[Proof of \Cref{thm:tribes-lower-bound} assuming \Cref{lemma:number-of-revealed-tribes,lemma:from-tribes-to-and,lemma:cost-of-seeing-a-variable}.]
    We first show that $\optavg_0 = O(2^w)$. Consider the algorithm that evaluates the $2^w$ tribes one by one. For each tribe, the algorithm queries the $w$ variables in increasing order of costs and stops whenever a zero is revealed. The probability of querying the cost-$i$ variable is $2^{-(i-1)}$, so the expected cost on each tribe is
    \[
        \sum_{i=1}^{w}2^{-(i-1)}\cdot i \le 4.
    \]
    This implies $\optavg_0 \le 4\cdot 2^w = O(2^w)$.

    For the lower bound part, let $\calA$ be an $\eps_0$-error algorithm computing the width-$w$ $\Tribes$ instance with an expected cost of $C$. By \Cref{lemma:number-of-revealed-tribes}, for some $p > 0$, $\calA$ reveals at least $p \cdot 2^w$ tribes in expectation. Then, by \Cref{lemma:from-tribes-to-and}, there is an algorithm for the $w$-variable $\AND$ instance that reveals at least one of the $w$ variables with probability $p$ and has an expected cost of $C/2^w$. By \Cref{lemma:cost-of-seeing-a-variable}, we must have $C/2^w \ge \Omega(w)$, which implies $C = \Omega(2^w \cdot w)$. 
\end{proof}

\subsection{Proof of \Cref{lemma:number-of-revealed-tribes}}
We start with \Cref{lemma:number-of-revealed-tribes}, which states that any algorithm that computes the $\Tribes$ function with a low error must reveal a constant fraction of the tribes. 

\numberofrevealedtribes*

Note that directly applying the OSSS inequality (\Cref{thm:osss-JZ-version}) would give a lower bound only on the number of revealed \emph{variables}, and not on the number of revealed \emph{tribes}. Our workaround is to transform an algorithm $\calA$---which reveals only a few tribes---into a ``canonical form'' algorithm $\calA'$, so that the number of revealed variables in $\calA'$ is comparable to that of revealed tribes in $\calA$. Furthermore, $\calA'$ has the same or smaller error probability. Then, applying the OSSS inequality to $\calA'$ would give the desired lower bound on the number of tribes revealed by $\calA$.

\begin{proof}
    Suppose that algorithm $\calA$ solves the width-$w$ $\Tribes$ instance with an error probability $\le \eps_0$ while revealing $m$ tribes in expectation. We will first derive another algorithm $\calA'$ that computes $\Tribes$ up to an error probability $\le \eps_0$ and queries at most $2m$ variables in expectation. Then, we apply the OSSS inequality to show that $2m \ge \Omega(2^w)$, which implies the desired lower bound $m \ge \Omega(2^w)$.
    
    \paragraph{Translating number of tribes to query complexity.} We will construct an alternative algorithm $\calA'$ for $\Tribes$. $\calA'$ works in a setting where the variables are not associated with costs, i.e., only the number of queries matters. Furthermore, our construction ensures that $\calA'$ satisfies the following two properties:
    \begin{itemize}
        \item \textbf{Property 1.} When $\calA'$ terminates, it computes the conditional expectation of $f$ given the revealed variables, and outputs the value rounded to $\zo$. Formally, letting $\pi$ denote the restriction formed by the observed variables, the output of $\calA'$ is always $\1{\ex{\bx\sim\zo^n}{f_{\pi}(\bx)} \ge 1/2}$.
        \item \textbf{Property 2.} $\calA'$ never queries a variable $x_{i,j}$ if, for some $j' \in [w]$, $x_{i,j'}$ is already queried and known to take value $0$.
    \end{itemize}
    Intuitively, Property~1 requires that the algorithm always makes the Bayes-optimal prediction. Property~2 prevents the algorithm from making an unnecessary query to a variable $x_{i,j}$ in the $i$-th tribe when that tribe is known to take value $0$ (due to the previous revelation of $x_{i,j'} = 0$).

    Formally, $\calA'$ is defined as follows:
    \begin{itemize}
        \item $\calA'$ simulates algorithm $\calA$ on a width-$w$ $\Tribes$ instance in \Cref{def:tribes-instance} by randomly choosing the costs $(c_{i,j})_{i \in [2^w], j \in [w]}$.
        \item Whenever $\calA$ reveals a variable $x_{i,j}$, $\calA'$ checks whether there exists another variable $x_{i,j'}$ in the same tribe that has already been revealed as a zero. If so, $\calA'$ draws a random bit and feeds that random bit to $\calA$ as the value of $x_{i,j}$; otherwise, $\calA'$ queries $x_{i,j}$ and forwards the value to $\calA$.
        \item When $\calA$ decides to terminate and output an answer, $\calA'$ outputs an answer according to Property~1.
    \end{itemize}

    In the following, we show that $\calA'$ queries $\le 2m$ variables in expectation, and has an error probability $\le \eps_0$.
    
    \paragraph{Upper bound the number of queries.} We first analyze the expected number of queries made by $\calA'$. Fix $i \in [2^w]$. We note that, conditioning on the event that $\calA$ reveals the $i$-th tribe, the conditional distribution of $(x_{i,1}, x_{i,2}, \ldots, x_{i,w})$ is still uniform over $\zo^w$. Then, by our construction of $\calA'$, $\calA'$ queries at least $j$ variables in the $i$-th tribe only if the first $j-1$ variables that get revealed are all ones, which happens with probability $2^{-(j-1)}$. Formally, we have
    \[
        \pr{}{\calA'\text{ queries at least }j\text{ variables in tribe }i}
    \le \pr{}{\calA\text{ reveals tribe }i}\cdot 2^{-(j-1)}.
    \]
    The expected number of variables in the $i$-th tribe queried by $\calA'$ is then given by
    \begin{align*}
        &~\ex{}{\sum_{j=1}^{w}\1{\calA'\text{ queries at least }j\text{ variables in tribe }i}}\\
    =   &~\sum_{j=1}^{w}\pr{}{\calA'\text{ queries at least }j\text{ variables in tribe }i}\\
    \le &~\sum_{j=1}^{w}\pr{}{\calA\text{ reveals tribe }i}\cdot 2^{-(j-1)}\\
    \le &~2\pr{}{\calA\text{ reveals tribe }i}.
    \end{align*}

    Therefore, the expected number of queries made by $\calA'$ is upper bounded by
    \[
        2\sum_{i=1}^{2^w}\pr{}{\calA\text{ reveals tribe }i}
    =   2\ex{}{\sum_{i=1}^{2^w}\1{\calA\text{ reveals tribe }i}}
    =   2m.
    \]
    
    \paragraph{Upper bound the error probability.} Next, we show that the error probability of $\calA'$ is at most $\eps_0$. To this end, we consider the following modified version of $\calA'$, denoted by $\calA''$:
    \begin{itemize}
        \item $\calA''$ simulates $\calA$ almost in the same way as $\calA'$ does. The only change is that, whenever $\calA$ reveals a variable $x_{i,j}$, $\calA''$ always queries that variable and forwards its value to $\calA$. (In contrast, $\calA'$ would feed a random bit to $\calA$ in some cases.)
        \item When $\calA$ terminates, $\calA''$ outputs an answer according to Property~1, i.e., it outputs the conditional expectation of $f$ (given the observed inputs) rounded to $\{0, 1\}$.
    \end{itemize}
    Compared with algorithm $\calA'$, $\calA''$ satisfies Property~1 but not Property~2.

    We first note that the error probability of $\calA''$ is never higher than that of $\calA$. When $\calA$ terminates, $\calA''$ and $\calA$ have observed the same subset of variables, which induce the same restriction $\bpi$. Conditioning on this event, all the unobserved variables are still uniformly distributed. Then, predicting $0$ leads to a conditional error probability of $\pr{\bx\sim\zo^n}{f_{\bpi}(\bx) \ne 0} = \ex{\bx\sim\zo^n}{f_{\bpi}(\bx)}$, while predicting $1$ leads to a conditional error of $1 - \ex{\bx\sim\zo^n}{f_{\bpi}(\bx)}$. Then, by predicting $\1{\ex{\bx\sim\zo^n}{f_{\bpi}(\bx)} \ge 1/2}$, algorithm $\calA''$ always has a lower or equal conditional error probability than $\calA$ does. Applying the law of total probability shows $\error_{\calA''}(f) \le \error_{\calA}(f)$.

    It remains to show that $\error_{\calA'}(f) \le \error_{\calA''}(f)$. We will show that these two sides are equal by coupling the two algorithms carefully. Suppose that both $\calA'$ and $\calA''$ simulate $\calA$ with the same randomness in $\calA$ and in the costs. Furthermore, we ``defer'' the randomness in the input $\bx \in \zo^n$ by 
    realizing  each input bit only when it gets queried by an algorithm. Whenever $\calA$ reveals a variable $x_{i,j}$, there are two cases:
    \begin{itemize}
        \item \textbf{Case 1: Some $x_{i,j'} = 0$ in the same tribe is already revealed.} Recall that $\calA'$ would feed a random bit to $\calA$ in this case, while $\calA''$ would query $x_{i,j}$ and feed the actual value to $\calA$. Note that, before $\calA''$ queries $x_{i,j}$, the input bit is uniformly distributed among $\{0, 1\}$. Thus, we may couple the two executions, so that the random bit chosen by $\calA'$ is always equal to the random realization of $x_{i,j}$ in $\calA''$.
        \item \textbf{Case 2: No variable $x_{i,j'}$ is revealed to be zero.} In this case, both $\calA'$ and $\calA''$ would actually query $x_{i,j}$. We couple the two executions such that the realization of $x_{i,j}$ are the same.
    \end{itemize}
    Note that this coupling ensures that $\calA$ follows the same computation path in both $\calA'$ and $\calA''$.

    At the end of the simulations, let $\bX', \bX'' \subseteq [2^w] \times [w]$ denote the indices of the variables that are queried by $\calA'$ and $\calA''$, respectively. By our coupling, we always have $\bX' \subseteq \bX''$. Let $\bpi'$ and $\bpi''$ denote the restrictions naturally induced by the variables observed by $\calA'$ and $\calA''$, respectively. Note that $\calA'$ and $\calA''$ output $\1{\ex{\bx\sim\zo^n}{f_{\bpi'}(\bx)} \ge 1/2}$ and $\1{\ex{\bx\sim\zo^n}{f_{\bpi''}(\bx)} \ge 1/2}$, respectively.

    We then couple the realization of the remaining randomness in $\bx$ such that, for every $(i, j) \in ([2^w] \times [w]) \setminus \bX''$, the realization of $x_{i,j}$ is the same for both $\calA'$ and $\calA''$. At this point, the input $\bx$ is determined in both simulations, and we let $\bx'$ and $\bx''$ denote the two realizations for clarity. We will verify the following two facts:
    \begin{itemize}
        \item $f(\bx') = f(\bx'')$, i.e., $f(\bx)$ takes the same value in both $\calA'$ and $\calA''$.
        \item The outputs of $\calA'$ and $\calA''$ are equal.
    \end{itemize}
    Assuming the above, we immediately have $\error_{\calA'}(f) = \error_{\calA''}(f)$, since
    \[
        \error_{\calA'}(f) = \pr{}{f(\bx') \ne\text{output of }\calA'}
    \quad\text{and}\quad
        \error_{\calA''}(f) = \pr{}{f(\bx'') \ne\text{output of }\calA''}.
    \]

    To verify the first fact, we note that $\bx'$ and $\bx''$ might differ only on coordinates
    $(i, j) \in \bX'' \setminus \bX'$. By definition of $\calA'$, for every such coordinate $(i, j)$, there is another variable in the $i$-th tribe that was revealed as a zero in both $\calA'$ and $\calA''$. Formally, there must exist $j' \in [w]$ such that $x'_{i,j'} = x''_{i,j} = 0$. Then, the difference $x'_{i,j} \ne x''_{i,j}$ would not matter, since the $i$-th tribe would take value $0$ in either case.
    
    To verify the second fact, for $\bx \in \zo^n = \zo^{2^w \cdot w}$, let $g_i(\bx) \coloneqq \bigwedge_{j=1}^{w}x_{i,j}$ denote the $i$-th tribe. For $x \in \zo^n$ and restriction $\pi$, we write $x \agrees \pi$ if $x$ is consistent with the restriction. 
    Then, for any generic restriction $\pi$, we have
    \[
        \ex{\bx\sim\zo^n}{f_{\pi}(\bx)}
    =   1 - \pr{\bx\sim\zo^n}{f(\bx) = 0 \mid \bx \agrees \pi}
    =   1 - \prod_{i=1}^{2^w}\left[1 - \pr{\bx\sim\zo^n}{g_i(\bx) = 1 \mid \bx \agrees \pi}\right].
    \]
    The second step above holds since, after conditioning on $\bx \agrees \pi$, $\bx$ always follows a product distribution.

    Therefore, it suffices to verify that, for every pair of restrictions $(\bpi', \bpi'')$ and every $i \in [2^w]$, it holds that
    \[
        \pr{\bx\sim\zo^n}{g_i(\bx) = 1 \mid \bx \agrees \bpi'}
    =   \pr{\bx\sim\zo^n}{g_i(\bx) = 1 \mid \bx \agrees \bpi''}.
    \]
    By our definition of $\calA'$ and $\calA''$, there are two possible cases:
    \begin{itemize}
        \item \textbf{Case 1: $\bX'$ and $\bX''$ agree on the $i$-th tribe.} In this case, $\bpi'$ and $\bpi''$ induce the same restriction on variables $x_{i,1}, x_{i,2}, \ldots, x_{i,w}$, so we have
        \[
            \pr{\bx\sim\zo^n}{g_i(\bx) = 1 \mid \bx \agrees \bpi'}
        =   \pr{\bx\sim\zo^n}{g_i(\bx) = 1 \mid \bx \agrees \bpi''}.
        \]
        \item \textbf{Case 2: For some $j \in [w]$, $(i, j) \in \bX'' \setminus \bX'$.} Again, this can only happen if, for some $j' \in [w]$, both $\bpi'$ and $\bpi''$ contain the restriction $x_{i,j'} = 0$, in which case we have
        \[
            \pr{\bx\sim\zo^n}{g_i(\bx) = 1 \mid \bx \agrees \bpi'}
        =   \pr{\bx\sim\zo^n}{g_i(\bx) = 1 \mid \bx \agrees \bpi''}
        =   0.
        \]
    \end{itemize}

    Therefore, we conclude that
    \[
        \error_{\calA'}(f) = \error_{\calA''}(f) \le \error_{\calA}(f) \le \eps_0,
    \]
    and $\calA'$ queries at most $2m$ variables in expectation.

    \paragraph{Lower bound on the query complexity of $\Tribes$.} This is a consequence of the OSSS inequality (\Cref{thm:osss-JZ-version}). In the width-$w$ $\Tribes$ function $f$, we have
    \[
        \pr{\bx\sim\zo^n}{f(\bx) = 0}
    =   \left(1 - \frac{1}{2^w}\right)^{2^w}
    \in \left[\frac{1}{4}, \frac{1}{e}\right].
    \]
    It follows that $\bias(f) \ge 1/4$. Furthermore, each variable is pivotal if and only if: (1) all the $2^w-1$ other tribes take value $0$; (2) all the $w-1$ other variables in the same tribe take value $1$. Thus, the influence of every variable $x_{i,j}$ is given by
    \[
        \Inf_{i,j}[f]
    =   \left(1 - \frac{1}{2^w}\right)^{2^w-1} \cdot \frac{1}{2^{w-1}}
    \le \frac{1}{2}\cdot\frac{2}{2^w}
    =   \frac{1}{2^w}.
    \]
    Then, for every $\eps_0$-error algorithm $\calA$, \Cref{thm:osss-JZ-version} gives
    \[
        \frac{1}{4} - \eps_0
    \le \bias(f) - \error_f(\calA)
    \le \sum_{i=1}^{2^w}\sum_{j=1}^{w}\delta_{i,j}(\calA) \cdot \Inf_{i,j}[f]
    \le \frac{1}{2^w}\sum_{i=1}^{2^w}\sum_{j=1}^{w}\delta_{i,j}(\calA).
    \]
    Therefore, the expected number of queries made by $\calA$, $\sum_{i=1}^{2^w}\sum_{j=1}^{w}\delta_{i,j}(\calA)$, is at least $(1/4 - \eps_0)\cdot 2^w$. Applying this to algorithm $\calA'$ constructed above gives $m \ge \frac{1/4 - \eps_0}{2}\cdot 2^w = \Omega(2^w)$.
\end{proof}

\subsection{Proof of \Cref{lemma:from-tribes-to-and}}
Next, we turn to \Cref{lemma:from-tribes-to-and}, which transforms an algorithm computing the width-$w$ $\Tribes$ function into one for the $w$-variable $\AND$ function. This is done by ``planting'' the $\AND$ instance as one of the $2^w$ tribes in the $\Tribes$ instance. The transformation ensures that the resulting algorithm reveals at least one of the $w$ variables with a sufficiently high probability, while incurring a low cost in expectation.

\fromtribestoand*

\begin{proof}
    Let $\calA$ denote the algorithm for width-$w$ $\Tribes$. In the following, we construct another algorithm, denoted by $\calA'$, for the $w$-variable $\AND$ function.
    \begin{itemize}
        \item Draw $\istar \in [2^w]$ uniformly at random.
        \item Generate a width-$w$ $\Tribes$ instance by sampling, independently for each $i \in [2^w] \setminus \{\istar\}$, a uniformly random permutation $(c_{i,1}, c_{i,2}, \ldots, c_{i,w})$ of $[w]$, and input bits $(x_{i,1}, x_{i,2}, \ldots, x_{i,w}) \in \zo^w$.
        \item Simulate $\calA$ on a width-$w$ $\Tribes$ instance with the $\istar$-th tribe being the actual $\AND$ instance. Formally, when $\calA$ increments the investment $\theta_{i,j}$ for some $i \in [2^w]\setminus\{\istar\}$ and $j \in [w]$, $\calA'$ checks whether $\theta_{i,j} \ge c_{i,j}$ holds. If so, $\calA'$ reveals the value of $x_{i,j}$ to $\calA$. When $\calA$ increments $\theta_{\istar, j}$ for some $j \in [w]$, $\calA'$ increases the investment in the $j$-th variable to $\theta_{\istar, j}$ in the $\AND$ instance (that $\calA'$ is solving).
    \end{itemize}

    By construction of $\calA'$, $\calA$ runs on a random width-$w$ $\Tribes$ instance. For each $i \in [2^w]$, let $p_i$ denote the probability that the $i$-th tribe is revealed by $\calA$, and let $C_i$ denote the expected cost that $\calA$ invests in the $w$ variables in the $i$-th tribe. Note that, conditioning on the event that $\istar = i$, the probability that algorithm $\calA'$ reveals at least one of the $w$ variables (in the $\AND$ instance) is exactly $p_i$. Furthermore, the expected cost incurred by $\calA'$ (in the $\AND$ instance) is exactly $C_i$.
    
    By our assumptions on $\calA$ and the linearity of expectation, we have
    \[
        \sum_{i=1}^{2^w}p_i \ge \alpha \cdot 2^w
    \quad\text{and}\quad
        \sum_{i=1}^{2^w}C_i \le C.
    \]
    Since $\istar$ is drawn uniformly at random from $[2^w]$, the overall probability that $\calA'$ reveals at least one of the $w$ variables is
    \[
        \frac{1}{2^w}\sum_{i=1}^{2^w}p_i
    \ge 2^{-w}\cdot (\alpha \cdot 2^w)
    =   \alpha,
    \]
    and the expected cost of $\calA'$ is
    \[
        \frac{1}{2^w}\sum_{i=1}^{2^w}C_i
    =   C / 2^w.
    \]
    Therefore, $\calA'$ would be the algorithm that the lemma asks for.
\end{proof}

\subsection{Proof of \Cref{lemma:cost-of-seeing-a-variable}}
\costofseeingavariable*
The proof follows the same idea as the one for \Cref{prop:and-lower-bound}: We simulate the algorithm on an $\AND$ instance in which every cost is high, so that none of the variables can be revealed. Then, we argue that the simulation coincides with the execution of the algorithm on an actual instance from \Cref{def:and-instance}. The proof needs to be slightly modified for two reasons: (1) we can no longer assume that the algorithm is deterministic; (2) The algorithm is only assumed to reveal a variable with probability $\Omega(1)$, rather than having a zero error for computing the $\AND$ function.

\begin{proof}
    Let $\calA$ denote the hypothetical algorithm for the $\AND$ instance. Set $m = \frac{p}{2}\cdot n$. We simulate $\calA$ on an alternative instance---in which the function is the $\AND$ over $n$ variables, each of which has a cost of $c_i = n$---until either one of the following two happens:
    \begin{itemize}
        \item \textbf{Case 1:} $\calA$ terminates voluntarily.
        \item \textbf{Case 2:} $\calA$ attempts to increase the total investment beyond $m$.
    \end{itemize}
    Let $\calA^{(1)}$ denote the simulation above. For each $i \in [n]$, let $a_i$ denote the total investment in variable $x_i$ when $\calA^{(1)}$ terminates. Since $\calA$ can be randomized in general, $a_1, a_2, \ldots, a_n$ are random variables that depend on the internal randomness of $\calA^{(1)}$. Nevertheless, it always holds that $\sum_{i=1}^{n}a_i \le m$; otherwise the simulation should have been terminated earlier by Case~2. Furthermore, since $m = (p/2) \cdot n \le n / 2 < n$, none of the $n$ variables gets revealed in simulation $\calA^{(1)}$.
    
    Let $\calA^{(2)}$ denote the simulation of $\calA$---using the same internal randomness as in $\calA^{(1)}$---on an actual $\AND$ instance from \Cref{def:and-instance}. Then, the two simulations $\calA^{(1)}$ and $\calA^{(2)}$ agree except with probability at most
	\[
		\sum_{i=1}^{n}\pr{c}{c_i \le a_i}
	\le \sum_{i=1}^{n}\frac{a_i}{n}
	\le \frac{m}{n}
	=	\frac{p}{2}.
	\]
    Moreover, since $\calA^{(2)}$ actually runs on an $\AND$ instance from \Cref{def:and-instance}, our assumption on $\calA^{(2)}$ implies that it reveals at least one variable with probability at least $p$.

    Let $E_1$ denote the event that $\calA^{(1)}$ and $\calA^{(2)}$ agree, and $E_2$ denote the event that $\calA^{(2)}$ eventually reveals a variable. The analysis above gives
    \[
        \pr{}{\overline{E_1}} \le \frac{p}{2}
    \quad\text{and}\quad
        \pr{}{E_2} \ge p,
    \]
    which implies
    \[
        \pr{}{E_1 \wedge E_2}
    =	\pr{}{E_2} - \pr{}{E_2 \wedge \overline{E_1}}
    \ge \pr{}{E_2} - \pr{}{\overline{E_1}}
    \ge p - \frac{p}{2}
    =	\frac{p}{2}.
    \]

    By our definition of the simulations, if both $E_1$ and $E_2$ happen, $\calA^{(2)}$ must reveal one of the $n$ variables after the total investment reaches $m$. Therefore, the expected cost of $\calA^{(2)}$ is at least
	\[
		\pr{}{E_1 \wedge E_2}\cdot m
	\ge \frac{p}{2} \cdot \frac{p}{2}n
	=	\frac{p^2}{4}\cdot n.
	\]
	In other words, the claim holds for $\alpha = p^2/4$.
\end{proof}

\section*{Acknowledgements}
R.R.~is supported by the NSF TRIPODS program (award DMS-2022448) and CCF-2310818.

\appendix

\bibliography{allrefs}
\bibliographystyle{alpha}

\end{document}